\documentclass[a4paper,fleqn]{cas-sc}

\usepackage[numbers]{natbib}

\usepackage{amsthm}
\usepackage{amssymb,bm}
\usepackage{amsmath}  
\usepackage{url}

\usepackage{graphicx}
\usepackage{subfigure}
\usepackage{setspace}

\usepackage{arydshln}

\usepackage{threeparttable} 
\usepackage[section]{placeins}
\usepackage{floatrow}
\newtheorem{theorem}{Theorem} 
\newtheorem{lemma}{Lemma} 
\newtheorem{remark}{Remark} 
\newtheorem{definition}{Definition} 

\newtheorem{example}{Example}
\newcommand{\RNum}[1]{\uppercase\expandafter{\romannumeral #1\relax}}

\def\tsc#1{\csdef{#1}{\textsc{\lowercase{#1}}\xspace}}
\tsc{WGM}
\tsc{QE}
\tsc{EP}
\tsc{PMS}
\tsc{BEC}
\tsc{DE}

\begin{document}

\let\WriteBookmarks\relax
\def\floatpagepagefraction{1}
\def\textpagefraction{.001}
\shorttitle{Unified framework for nonlinear grey model}
\shortauthors{Lu Yang et~al.}

\title [mode = title]{On unified framework for nonlinear grey system models: an integro-differential equation perspective}


\author[1]{Lu Yang}[style=chinese]
\ead{yang_lu@nuaa.edu.cn}
\credit{Conceptualization, Methodology, Software, Writing - Original draft}

\address[1]{College of Economics and Management, Nanjing University of Aeronautics and Astronautics, Nanjing 211106, PR China}

\author[1]{Naiming Xie}[style=chinese]
\cormark[1]
\ead{xienaiming@nuaa.edu.cn}

\credit{Funding acquisition, Supervision, Writing - review \& editing}

\author[1]{Baolei Wei}[style=chinese]
\ead{weibl@nuaa.edu.cn}
\credit{Methodology, Validation, Writing - review \& editing}

\author[1]{Xiaolei Wang}[style=chinese]
\ead{wangxiaolei0721@163.com}
\credit{Visualization, Writing - review \& editing}




\cortext[cor1]{Corresponding Author}


\begin{abstract}
Nonlinear grey system models, serving to time series forecasting, are extensively used in diverse areas of science and engineering. However, most research concerns improving classical models and developing novel models, relatively limited attention has been paid to the relationship among diverse models and the modelling mechanism. The current paper proposes a unified framework and reconstructs the unified model from an integro-differential equation perspective. First, we propose a methodological framework that subsumes various nonlinear grey system models as special cases, providing a cumulative sum series-orientated modelling paradigm. Then, by introducing an integral operator, the unified model is reduced to an equivalent integro-differential equation; on this basis, the structural parameters and initial value are estimated simultaneously via the integral matching approach. The modelling procedure comparison further indicates that the integral matching-based integro-differential equation provides a direct modelling paradigm. Next, large-scale Monte Carlo simulations are conducted to compare the finite sample performance, and the results show that the reduced model has higher accuracy and robustness to noise. Applications of forecasting the municipal sewage discharge and water consumption in the Yangtze River Delta of China further illustrate the effectiveness of the reconstructed nonlinear grey models.
\end{abstract}



\begin{keywords}
nonlinear grey system models \sep cumulative sum operator \sep integro-differential equation \sep integral matching  \sep municipal sewage discharge
\end{keywords}

\maketitle

\section{Introduction}

\setlength\baselineskip{15pt}

State space models, a powerful framework for time series, are extensively used to fit the measurement data and then forecast the evolution of dynamic systems.
The main feature of state space models is that depicts the underlying process describing dynamic law of the system in terms of states \citep{king2012review}.
There exists a diversity of state space models, such as transfer function \citep{stoica2000mimo}, exponential smoothing \citep{hyndman2002state}, and dynamic linear regression \citep{young2011gauss}.
In nature, grey system models proposed by \citet{deng1984grey} belong to the state space systems.
One of the key issues of the grey system models is identifying the governing equations from measurement data for physical understanding, forecasting and controlling.
Over the past four decades, a range of grey system models, including linear and nonlinear systems, have been emerging to solve a class of time series forecasting problems.
It is worth noting that rather than fitting the original time series, however, grey system models visualize and identify the pattern hidden in the original time series by utilizing cumulative sum (Cusum) operator, distinguishing themselves from other classical state space models.
Recently, by employing the integral matching approach, we explained the mechanism of Cusum operator from the perspectives of mathematical analysis and parameter estimations \citep{wei2020understanding} .
Then, on this basis, a unified modelling paradigm for linear grey forecasting models is proposed to link grey system models with dynamic data analysis \citep{wei2020unified1,wei2020unified}.

Nonlinear grey system models, aimed at describing more complex systems, have broader applicability than linear ones, having attracted considerable attention in the grey system community. Up to now, most research is focused on improving the classical models and developing novel ones, which enriches nonlinear grey system model families.
As an emerging method, however, there remain several challenges needed to be addressed.
First, there exist many nonlinear grey system models with different forms, making it difficult for researchers to perform property analysis. Hence it is required to form a unified methodological framework, fully illustrating the modelling paradigm.
The distinguishing feature of nonlinear grey models is fitting the Cusum series, but the reason why nonlinear grey systems models use the Cusum operator has not been well explained. So, is it possible that fitting the original time series instead of Cusum one but achieving the same forecasts?
If so, the mechanism of nonlinear grey models will be revealed, making the modelling results easier to understand and explain.
Besides, both parameter estimation and initial value selection affect the forecasting results.
Current research tends to use a two-step method: estimating the structural parameters and subsequently selecting the initial condition via a given strategy (see \cite{wei2020unified1} for the most used three strategies).
The inconsistent objective functions in two separate stages may introduce extra errors and degrade outcomes significantly \cite{wei2019data}, so it is necessary to develop a one-step parameter estimation method.

Therefore, in this work, we seek a unified framework for nonlinear grey system models and reconstruct the resultant unified form via an integro-differential equation.
The principal contributions are summarized as follows:
\begin{enumerate}
	\item[(1)] We propose a new methodological framework for nonlinear grey system models, which not only has the ability to unify the existing systems, but may also induce novel ones, providing the basis of nonlinear grey modelling paradigm, i.e., the Cusum series-orientated nonlinear differential equation models.

	\item[(2)] By introducing an integral operator, we reduce the unified representation to an integro-differential equation, allowing us to fit the original time series directly and analyse the mechanism of Cusum operator from both mathematical analysis and parameter estimation perspectives.

	\item[(3)] Based on the reduced integro-differential equation, a new one-step parameter estimation approach, integral matching, is introduced to estimate structural parameters and initial value simultaneously.
\end{enumerate}

The remaining parts are organized as follows.
Section \ref{sec2} reviews the existing nonlinear grey system models from the original to the extended ones.
Section \ref{sec3} presents the unified framework for nonlinear grey system models.
Section \ref{sec4} proposes the reduced reconstruction via an integro-differential equation.
Section \ref{sec5} conducts Monte Carlo simulations to evaluate the finite sample performance.
Section \ref{sec6} provides a real-world application and section \ref{sec7} concludes the work.


\section{Literature review} \label{sec2}

In this section, we review the existing nonlinear grey system models from two viewpoints:
the basic nonlinear grey model and the extended ones.

\subsection{The basic grey Verhulst model}
The grey Verhulst model (GVM(1,1)) \cite{deng1984grey}, where the first "1" denotes the order of derivative and the second "1" denotes the dimensional of variables, aims at fitting inverted U-shaped time series and lays the foundation for the development of nonlinear grey system models.

\begin{definition}
	\label{def1}
    \citep{wei2020unified1}
	For a time series with n samples $X(t) = \left\lbrace x(t_1),~x(t_2),\cdots,~x(t_n) \right\rbrace$, the Cusum series is defined as $Y(t) = \left\lbrace y(t_1),~y(t_2),\cdots,~y(t_n) \right\rbrace$ for $y(t_k) = \sum_{i = 1}^{k} h_i x(t_k)$, where $h_1 = 1$ and $h_k = t_k - t_{k-1}$ for $k\ge2$.
\end{definition}

The grey Verhulst model consists of the differential equation
\begin{equation} \label{verh}
	\frac{d}{dt}y(t) = ay(t) + b\left( y(t)\right) ^{2},~t\ge t_1
\end{equation}
and the corresponding discrete-time equation
\begin{equation}
	x(t_k) = a\left[ \lambda y(t_{k-1}) + (1- \lambda) y(t_{k}) \right] +
	b \left[ \lambda y(t_{k-1}) + (1- \lambda) y(t_{k}) \right] ^2,
	~k = 2,~3\cdots,~n
\end{equation}
where $\lambda \in [0,1]$ is referred to as a background coefficient whose value is always set to $0.5$.

The structural parameters can be easily estimated by the least squares criterion
\[
\begin{bmatrix}
	\hat{a}~~ \hat{b}
\end{bmatrix}^\top  =
	\mathop{\arg\min}\limits_{a,b} \left\|
	\mathbf{Y}-\mathbf{B} \begin{bmatrix} a~~b \end{bmatrix}^\top
	\right\|_2^2
	= {\left( {{\mathbf{B}^\top}\mathbf{B}} \right)^{ - 1}}{\mathbf{B}^\top}\mathbf{Y}
\]
where
\[
\mathbf{B}=
\begin{bmatrix}
	\frac{y({t_1})+y({t_2})}{2} & \left(\frac{y({t_1})+y({t_2})}{2} \right) ^2 \\
	\frac{y({t_2})+y({t_3})}{2} & \left(\frac{y({t_2})+y({t_3})}{2} \right) ^2 \\
    \vdots & \vdots \\
	\frac{y({t_{n-1}})+y({t_n})}{2} & \left(\frac{y({t_{n-1}})+y({t_n})}{2} \right) ^2 \\
\end{bmatrix},
~
\mathbf{Y}=
\begin{bmatrix}
	{x({t_2})}\\
	{x({t_3})}\\
	\vdots \\
	{x({t_n})}
\end{bmatrix}.
\]

Let the initial value be $y(t_1) = \eta_y$, the closed form solution of equation \eqref{verh} is
\[
\hat{y}(t) = 
	\left[-\frac{\hat {b}}{\hat{a}} + e^{-\hat{a}(t - t_0)} \cdot \left( \frac{1}{{\eta}_y} + 
	\frac{\hat{b}}{\hat{a}} \right)  \right] ^{-1}
\]
then the desired forecasts $\hat{X}$ can be calculated by the inverse Cusum operator derived from Definition \ref{def1}, that is,
\[
\begin{cases}
	\hat{x}(t_1) = \eta_y \\
	\hat{x}(t_k) = \frac{1}{h_k} \left( \hat{y}(t_k) - \hat{y}(t_{k-1})\right), ~k = 2,~3,\cdots,~n+r
\end{cases}
\]
where $r$ is the forecasting horizon.

Subsequently, a great deal of effort has been devoted to improve the accuracy and the existing research can be divided into the following three types: (i) improving the structural parameter estimates via intelligence algorithm \citep{shaikh_forecasting_2017,wang2019modelling}, weighted least squares \citep{tang_study_2020}, and parameter transformation \citep{evans2014alternative}, (ii) optimizing the initial value selection strategy \cite{chen2010combination}, and (iii) searching the optimal background coefficient \cite{duan2020grey}.

\subsection{Extended nonlinear grey system models}

Following the similar ideas and modelling procedures, the extended models utilize other nonlinear differential equations to fit Cusum series.
A rough classification of the nonlinear grey system models includes two categories: single-output and multi-output models, and we present the main extensions in Table \ref{review}.

\begin{table}[h]
	\centering
	\begin{threeparttable}
		\centering
		\scriptsize
		\caption{The coupled equations of continuous-time nonlinear grey system models}.
		\begin{spacing}{1.5}
			\setlength{\tabcolsep}{1.0mm}
			{
				\begin{tabular}{lllll}
					\toprule
					Type &	Model  & Differential equation	& Difference equation	& Ref\\
					\midrule
					single-output 	& GVM(1,1)	&  $\frac{d}{dt}y(t)=ay(t)+b\left( y(t) \right)^2 $
					& $x(t_k) =a\left(z(t_k)\right) + b\left(z(t_k)\right)^2$	& \cite{shaikh_forecasting_2017}\\
					& GGVM(1,1)	& $\frac{d}{dt}y(t)=ay(t)+b\left( y(t) \right)^2+c$
					& $x(t_k) =a\left(z(t_k)\right)+ b\left(z(t_k)\right)^2+c$ 	& \cite{zhou_grey_2020}
					\\
						& GRM(1,1)	& $\frac{d}{dt}y(t) = ay(t)+b\left( y(t) \right)^2+c(k-1)^\gamma+d$
					&$x(t_k)=az(t_k)+b\left(z(t_k) \right)^2+c(k-1)^\gamma+d $ & \cite{gao_novel_2021}
					\\
					& NGM(1,1,$\alpha$)	& $\frac{d}{dt}y(t)=a  \left( y(t)\right) ^\alpha  + b$
					& $x(t_k)=a\left(z(t_k) \right)^\alpha +b$	& \cite{wang_forecasting_2019}\\
					& NGBM(1,1)	& $\frac{d}{dt}y(t)=ay(t)+b\left( y(t)\right) ^{\gamma}$ 	& $x(t_k) =a\left(z(t_k)\right)+
					b\left(z(t_k)\right)^\gamma$ & \cite{chen_forecasting_2008}
					\\
					& GRBM(1,1)	& $\frac{d}{dt}y(t)=ay(t)\mathsf+b{{\left( y(t) \right)}^{\gamma }}+c$
					&	$x(t_k) =a\left( z(t_k)\right) + b\left(z(t_k)\right)^2+c$	& \cite{xiao_novel_2020} \\
					& NGBMC(1,n)	&  {$\frac{d}{dt} y_1(rp+t) = ay_1(rp+t) + \left(\sum\limits_{i=2}^{n} b_iy_i(t)+u \right)\left( y_1(rp+t)\right)^2$}
					& no explicit from	& \cite{yu_novel_2021} 
					\\
					& KRNGM(1,n)	& $\frac{d}{dt} y_\mathsf{out}(t)=ay_{\mathsf{out}}(t)+\bm \omega^\top \bm \phi(\bm{y}_\mathsf{in}(t))+u$
					&$x_{\mathsf{out}}(t_k)=az_\mathsf{out}(t_k)+ \bm \omega^\top \bm \phi(\bm{y}_\mathsf{in}(t_k))+u$	& \cite{ma_kernel-based_2018}\\
					multi-output
					& GLVM(1,2)
					& no explicit form
					& $\begin{cases}
						x_1(t_k) = a_1z_1(t_k)+ b_1\left( z_1(t_k)\right)^2 + c_1z_1(t_k)z_2(t_k)  \\
						x_2(t_k) = a_2z_2(t_k) + b_2 \left( z_2(t_k)\right) ^2 + c_2z_1(t_k)z_2(t_k)
					\end{cases}$
					& \cite{gatabazi_grey_2019} 
					\\
					\bottomrule
				\end{tabular}
			}
		\label{review}
			\begin{tablenotes}
				\footnotesize
				\item[1] $z_i(t_k) = \omega y_i(t_k)+ (1-\omega)y_i(t_{k-1})$
			\end{tablenotes}
		\end{spacing}
	\end{threeparttable}
\end{table}

Table \ref{review} shows that single-output models subsume single-variable and multi-variable models.
Single-variable models share a similar representation. For instance, when $c = 0$, GRBM(1,1) can reduce to NBGM(1,1) and, then, if $\gamma$ in NGBM(1,1) equals 2, then GVM(1,1) is obtained.
For multi-variable extensions, in particular, KRNGM(1,n) uses kernel tricks to introduce other variables.
Further ones lead to grey output models, but to our knowledge, research of multi-output models is sparse and mainly focused on the grey prey-predator system (GLVM(1,2)) concerning accuracy improvement \citep{wu_grey_2012,mao_impact_2015} and stability analysis \citep{gatabazi_grey_2019,wang_testing_2016}.

Another extension is hybrid models, including the grey time-delayed Verhulst model which employs a time-delay differential equation \citep{wang2013forecasting}, the F-NGBM(1,1) model which uses the Fourier series as residual correction tool \cite{chia-nan_improved_2015}, and the metabolic nonlinear grey–autoregressive integrated moving average model which couples NGM(1,1,$\alpha$) and ARIMA models \citep{wang2018forecasting}.

Similar to the research route of grey Verhulst model, a range of studies were implemented to optimize the aforementioned models, that is, the structural parameter, initial value and background coefficient optimization for the single-variable cases \citep{hsu_genetic_2010,yuan_safsa-_2020}, the multi-variable cases \citep{duman_estimation_2019}, and the multi-output cases \citep{gatabazi_fractional_2019}.

\section{Unified framework for nonlinear grey system models} \label{sec3}

In this section, we propose a unified framework for existing nonlinear grey system models. For a $d$-dimensional state vector $\bm{x}(t)$, the observations are sampled at time points $\left\{t_1, t_2, \cdots, t_n \right\}$ and arranged into the following matrix:
\[
\begin{bmatrix}
	\bm{x}^\top(t_1) \\
	\bm{x}^\top(t_2)  \\
	\vdots \\
	\bm{x}^\top(t_n)
\end{bmatrix}
=
\begin{bmatrix}
	x_1(t_1) & x_2(t_1)	 & \cdots & x_d(t_1) \\
	x_1(t_2) & x_2(t_2)	 & \cdots & x_d(t_2) \\
	\vdots   & \vdots    & \ddots & \vdots	 \\
	x_1(t_n) & x_2(t_n)	 & \cdots & x_d(t_n) \\
\end{bmatrix}.
\]
Then, the corresponding Cusum operator matrix is defined as
\[
\begin{bmatrix}
	\bm{y}^\top(t_1) \\
	\bm{y}^\top(t_2)  \\
	\vdots \\
	\bm{y}^\top(t_n)
\end{bmatrix}
=
\begin{bmatrix}
	y_1(t_1) & y_2(t_1)	 & \cdots & y_d(t_1) \\
	y_1(t_2) & y_2(t_2)	 & \cdots & y_d(t_2) \\
	\vdots   & \vdots    & \ddots & \vdots	 \\
	y_1(t_n) & y_2(t_n)	 & \cdots & y_d(t_n) \\
\end{bmatrix}.
\]
where ${\bm{y}}(t_i)= \sum_{i=1}^{k}h_i{\bm{x}}(t_i)$ for $h_1=1$ and $h_k=t_k-t_{k-1}$, $k\geq2$.

Nonlinear grey system models utilize first order nonlinear differential equations to describe the evolution of the Cusum variable.
Consider a unified representation
\begin{equation} \label{gm}
	\frac{d}{dt}\bm{y}(t) =  \bm{\theta}_{\text{L}} \bm{y}(t) + \bm{\theta}_{\text{N}} \bm{N}(\bm{y}(t)) +\bm{\beta},~\bm{y}(t_1)=\bm{\eta}_y,~ t\ge t_1
\end{equation}
where $\bm{y}(t)\in \mathbb{R}^d$ is the Cusum state, $\bm{N}\left(\bm{y}(t) \right):\mathbb{R}^d \rightarrow \mathbb{R}^p$ is a $p$-dimensional nonlinear vector function, $\bm{\eta}_y$ is the unknown initial value,
$\bm{\theta}_{\text{L}} \in \mathbb{R}^{d \times d}$, $\bm{\theta}_{\text{N}} \in \mathbb{R}^{d \times p}$ and $\bm{\beta}\in\mathbb{R}^d$ are the unknown structural parameters.

Note that, here, we explicitly indicate that the dynamics \eqref{gm} have both linear and nonlinear contributions for the convenience of parameter estimation.
Then, the two-step least squares is utilized to estimate the structural parameters and initial value in succession. In the first step, using the implicit Midpoint method gives the corresponding discrete-time equation
\begin{equation}\label{gm.disc}
	\frac{\bm{y}(t_k) - \bm{y}(t_{k-1})}{t_k - t_{k-1}} =
		\bm{x}(t_k) \approx
		 \bm{\theta}_{\text{L}}  \frac{\bm{y}(t_k) + \bm{y}(t_{k-1})}{2}
		+ \bm{\theta}_{\text{N}} \bm{N} \left( \frac{\bm{y}(t_k) + \bm{y}(t_{k-1})}{2} \right)
		+ \bm{\beta} + \bm{\epsilon}(k)
\end{equation}
where $\bm{\epsilon}(k)$ is the model error and $\frac{1}{2}\left(\bm{y}(t_k)+\bm{y}(t_{k-1}) \right)$ is referred to as background value in the grey system terminology.
By substituting $k = 2,3,\cdots,n$ into equation \eqref{gm.disc} and arranging the resulting $n-1$ algebraic equations into a matrix form, we have
\begin{equation}\label{gm.ls}
	\bm{X} = {\Theta}(\bm{y}){\Xi} + \Gamma
\end{equation}
where
\[
{\Xi} =
	\begingroup
	\renewcommand*{\arraystretch}{1.2}
	 \begin{bmatrix}
		\bm{\theta}^\top_{\text{L}} \\
		\bm{\theta}^\top_{\text{N}} \\
		\bm{\beta}^\top
	\end{bmatrix},
	\endgroup
~
\bm{X} =
	\begin{bmatrix}
		\bm{x}^\top(t_2) \\
		\bm{x}^\top(t_3)  \\
		\vdots \\
		\bm{x}^\top(t_n)
	\end{bmatrix},
~
{\Theta}(\bm{y}) =
	\begin{bmatrix}
		\frac{\bm{y}^\top(t_1)+\bm{y}^\top(t_2)}{2} &
		\bm{N}^\top \left(\frac{\bm{y}(t_1)+\bm{y}(t_2)}{2}\right) & 1\\
		\frac{\bm{y}^\top(t_2)+\bm{y}^\top(t_3)}{2} &
		\bm{N}^\top \left(\frac{\bm{y}(t_2)+\bm{y}(t_3)}{2}\right) & 1\\
		\vdots & \vdots & \vdots \\
		\frac{\bm{y}^\top(t_{n-1})+\bm{y}^\top(t_n)}{2} & \bm{N}^\top\left(\frac{\bm{y}(t_{n-1})+\bm{y}(t_n)}{2}\right) & 1
	\end{bmatrix},
~
{\Gamma} = \begin{bmatrix}
	\bm{\epsilon}^\top(2)\\
	\bm{\epsilon}^\top(3)\\
	\vdots \\
	\bm{\epsilon}^\top(n)\\
\end{bmatrix}.
\]

This then allows for the formulation of a regression problem to estimate the structural parameters:
\begin{equation}\label{eq2.5}
	\mathop{\min}_{{\Xi}} ~\mathcal{L}({\Xi}) = \Big\|  \bm{X} -
				{\Theta}(\bm{y}){\Xi} \Big\|^2_{\bm{\mathsf{F}}}
\end{equation}
where $\| \cdot \|_\mathsf{F}$ is the Frobenius norm.
Note, problem \eqref{eq2.5} is a linear least squares problem due to the separable parameters. Differentiating $\mathcal{L}({\Xi})$ with respect to ${\Xi}$ yields
\[
\frac{\partial}{\partial {\Xi}} \mathcal{L}({\Xi}) =
	\frac{\partial}{\partial {\Xi}} \mathsf{Tr} \left(\left[ \bm{X}-{\Theta}(\bm{y}){\Xi} \right]^\top \left[ \bm{X}-{\Theta}(\bm{y}){\Xi}\right]  \right) = -2\Theta^{\top}(\bm{y})\Xi + 2\Theta^{\top}(\bm{y})\Theta(\bm{y}){\Xi}
\]
where, setting the derivative to zero gives the least-squares estimates
\[
	{\hat{\Xi}} = \left( \Theta^{\top}(\bm{y})\Theta(\bm{y})\right)^{-1} \Theta^{\top}(\bm{y}) \bm{X}.
\]

In the second step, by substituting the estimated structural parameters into equation \eqref{gm} and solving the resultant differential equations, we have the solution (also termed as time response function) expressed as
\begin{equation}\label{gm.ini}
	\hat{\bm{{y}}}(t) = \bm{F}(\bm{\eta}_y,{\hat{\Xi}};t)
\end{equation}
where $\hat{\bm{y}}(t_1) = \bm{\eta}_y \in \mathbb{R}^d$ is unknown initial condition.
Note, the closed form solutions of equation \eqref{gm} are often not available due to the nonlinearity.
We can alternatively calculate the numerical solutions via numerical simulation scheme such as the Runge-Kutta algorithm \citep{mattheij2002ordinary}. It is obvious that $\bm{\eta}_y$ is crucial to the solution \eqref{gm.ini}, and three popular initial condition selection strategies are as follows:
\begin{enumerate}
	\item[(1)] the fixing first point strategy: $\bm{\hat{\eta}}_y$ is the solution to ${\bm{y}}(t_1) = \hat{\bm{{y}}}(t_1) = \bm{F}(\bm{\eta}_y,{\hat{\Xi}};t_1)$;
	\item[(2)] the fixing last point strategy: $\bm{\hat{\eta}}_y$ is the solution to ${\bm{y}}(t_n) = \hat{\bm{{y}}}(t_n) = \bm{F}(\bm{\eta}_y,{\hat{\Xi}};t_n)$;
	\item[(3)] the residual error correction strategy:
	 $\bm{\hat \eta}_y  = \mathop{\arg\min}_{\bm{\eta}_y}  \left\{
	 \sum_{i=1}^{n}\left\| \bm{F}(\bm{\eta}_y,{\hat{\Xi}};t_i)-\bm{y}(t_i) \right\|_2^2 \right\}$.
\end{enumerate}

Finally, substituting time points $\left\lbrace t_i \right\rbrace^{n+r}_{i=1} $ into equation \eqref{gm.ini} gives the fitting and forecasting values $\hat{\bm{y}}(t)$  of Cusum series and subsequently, by applying the inverse Cusum operator, the forecasts corresponding to the original time series are given through
\[
\hat{\bm{x}}(t_k) =
\begin{cases}
	\hat{\bm{y}}(t_1), & k =1,\\
	\frac{1}{h_k} \left(\hat{\bm{y}}(t_k)-\hat{\bm{y}}(t_{k-1}) \right) & k =2,3,\cdots,r.
\end{cases}
\]

To conclude, we present a unified representation for nonlinear grey system models and in the following, we further discuss the backwards compatibility of this unified, including the single-output and multi-output grey nonlinear models.

\begin{remark}\label{rm1}
	Let the dimension of the system be $d=1$.
	The unified model \eqref{gm} yields multiple families of single-output nonlinear grey system models with different $\bm{N}(y)$.
	\begin{enumerate}
		\item[(1)] Considering that  $\bm{N}(y)$ consists of polynomial term, the differential equation and the corresponding discrete-time equation can be written as
		\[
		\frac{d}{dt}y(t) = ay(t) + \sum_{\iota=1}^{p+1} b_{\iota} \left( y(t)\right) ^{\iota} + \beta
		\]
		and
		\[
		x(t_k) = a\frac{y(t_{k-1})+y(t_k)}{2} +
			\sum_{\iota=1}^{p+1} b_{\iota} \left( \frac{y(t_{k-1})+y(t_k)}{2}\right) ^{\iota} + \beta + \epsilon(t_k).
		\]

		\item[(2)] Considering that $\bm{N}(y)$ consists of power form, the differential equation and the corresponding discrete-time equation can be written as
		\[
		\frac{d}{dt}y(t) =  ay(t) + b\left( y(t)\right) ^\gamma + \beta
		\]
		and
		\[
		x(t_k) = a\frac{y(t_{k-1})+y(t_k)}{2} + b\left( \frac{y(t_{k-1})+y(t_k)}{2}\right) ^\gamma + \beta +\epsilon(t_k).
		\]
	\end{enumerate}
\end{remark}

Remark \ref{rm1} demonstrates that the unified model subsumes a number of nonlinear grey systems model families, although we only give two scenarios here.
Furthermore, the polynomial and power families cover all the single-variable models in Table \ref{review}. In practice, particularly, the power model takes $\gamma$ as a hyper-parameter which can be determined by optimization methods, such as particle swarm algorithm \cite{ding2021novel} and line search \cite{yang_integral_2021}.

\begin{remark}\label{rm2}
	Let the dimension of the state vector be $d \ge 2$.
	Similar to the extensions in Remark \ref{rm1}, if $\bm{N}(\bm{y})$ consists of polynomial term $\begin{bmatrix}
		\bm{y}^{P_2}&\bm{y}^{P_3}~\cdots
	\end{bmatrix}$, 
	the unified model yields a number of multi-output nonlinear grey models.
	Here, higher polynomials are denoted as $\bm{y}^{P_2},~\bm{y}^{P_3}$, etc., where $\bm{y}^{P_2}$ denotes the quadratic nonlinearities in $\bm{y}$: 
	$\begin{bmatrix}
		y_1^2 & y_1y_2 & \cdots & y_1y_d & y_2^2  & \cdots & y_d^2
	\end{bmatrix}$.
	
    Among them is the most simple form having the continuous- and corresponding discrete-time equation expressed as
	\[
	\begin{cases}
		\frac{d}{dt}y_1(t) = y_1(t)\left(a_{11}+\sum\limits_{i=1}^{d}b_{1i}y_i(t) \right) \\
		\frac{d}{dt}y_2(t) = y_2(t)\left(a_{21}+\sum\limits_{i=1}^{d}b_{2i}y_i(t) \right) \\
		~~~~~~~~\vdots \\
		\frac{d}{dt}y_d(t) = y_d(t)\left(a_{d1}+\sum\limits_{i=1}^{d}b_{di}y_i(t) \right)
	\end{cases}
	and~~
	\begin{cases}
		x_1(t_k) = z_1(t)\left( a_{11}+\sum\limits_{i=1}^{d}b_{1i} z_i(t) \right) \\
		x_2(t_k) = z_2(t)\left( a_{21}+\sum\limits_{i=1}^{d}b_{2i} z_i(t) \right) \\
		~~~~~~~~\vdots \\
		x_d(t_k) = z_d(t)\left( a_{d1}+\sum\limits_{i=1}^{d}b_{di} z_i(t) \right) \\
	\end{cases}
	\]
	where $z_i(t_k) = \frac{1}{2}y_i(t_{k-1})+\frac{1}{2}y_i(t_{k}),~i=1,2,\cdots,d$.
\end{remark}

Remark \ref{rm2} shows that the unified model can yield the high-dimensional system that are employed to describe interactions between species and, additionally, if $d=2$, then the unified model is actually the classical grey Lotka-Volterra model \citep{gatabazi_grey_2019}.
Furthermore, similar to the expansions in Remark \ref{rm1}, the unified model could deduce some other novel multi-output models.

\begin{remark}\label{rm3}
	If there exist forcing terms, then the unified model can be extended to
	\begin{equation}\label{gm.u}
	\frac{d}{dt}\bm{y}(t) =
		\bm{\theta}_{\text{L}}\bm{y}(t) + \bm{\theta}_{\text{N}}\bm{N}(\bm{y}(t)) + \bm{\theta}_\mathsf{F}\bm{u}(t) + \bm{\beta},
		~
	\bm{y}(t_1) = \bm{\eta}_y
	\end{equation}
	where $\bm{u}(t)\in \mathbb{R}^\ell$ is a known input vector which is independent of $\bm{y}(t)$, $\bm{\eta}_y$ is unknown initial value, $\bm{\theta}_{\text{L}}\in \mathbb{R}^{d\times d},~\bm{\theta}_{\text{N}}\in \mathbb{R}^{d\times p},~\bm{\theta}_\mathsf{F}\in \mathbb{R}^{d\times\ell}$ and $\bm{\beta}\in \mathbb{R}^{d}$ are structural parameters.

	In light of equation \eqref{gm.disc}, we obtain the corresponding discrete-time equation
	\[
	\bm{x}(t_k) \approx \bm{\theta}_{\text{L}} \frac{\bm{y}(t_k) + \bm{y}(t_{k-1})}{2}
	+ \bm{\theta}_{\text{N}} \bm{N} \left( \frac{\bm{y}(t_k) + \bm{y}(t_{k-1})}{2} \right)
	+ \bm{\theta}_\mathsf{F} \frac{\bm{u}(t_k) + \bm{u}(t_{k-1})}{2}
	+ \bm{\beta} + \bm{\epsilon}(k).
	\]
\end{remark}

Remark \ref{rm3} gives an extension principle of the unified nonlinear grey system model when introducing forcing terms to nonlinear grey models.
For instance, considering the single-output scenario, if $\bm{u}(t)$ consists of other state variables, then we can obtain a number of multi-variable models, including the NGBMC(1,n) and KRGBM(1,n) in Table \ref{review}.

Remarks \ref{rm1}-\ref{rm3} show that the unified model has the capacity to represent classical grey nonlinear models, including but not limited to the existing single-output, and multi-output models.
The extension of the above parameter estimation procedures to equation \eqref{gm.u} is straightforward and tedious. In order to be focused, we only discuss the model \eqref{gm} in the following.

\section{Reconstruction of nonlinear grey models with an integro-differential equation}\label{sec4}

In this section, we simplify the unified nonlinear grey model into a reduced-order integro-differential equation;
then, by using the integral matching approach which consists of an integral operator and the least squares \citep{dattner2017modelling}, we estimate the structural parameters and initial conditions simultaneously; further, we compare the modelling procedures of nonlinear grey models with those of the integro-differential equation-based ones.

\subsection{Reduced-order integro-differential equation model}
\begin{theorem}\label{lem1}
	Let
	\begin{equation} \label{im}
		\bm{y}(t) = \bm{\eta}_y + \int_{t_1}^{t}\bm{x}(\tau)d\tau,~t\ge t_1
	\end{equation}
	where $\bm{\eta}_y \in \mathbb{R}^d$ is a real vector,
	the aforementioned unified representation \eqref{gm} is equivalent to an integro-differential equation expressed as
	\begin{equation} \label{im.red}
	\frac{d}{dt}\bm{x}(t) =
		\bm{\theta}_{\text{L}}\bm{x}(t) + \bm{\theta}_{\text{N}}\bm{x}(t)  \frac{d}{dt}\bm{N}\left(\bm{\eta}_y+\int_{t_1}^{t}\bm{x}(\tau)d\tau\right),
				~\bm{x}(t_1)=\bm{\eta}_x,~t\ge t_1
	\end{equation}
	where $\bm{\eta}_x = \bm{\theta}_{\text{L}}\bm{\eta}_y + \bm{\theta}_{\text{N}}\bm{N}(\bm{\eta}_y) + \bm{\beta}$.
\end{theorem}

\begin{proof}
	To begin with the necessity, that is, equation \eqref{gm} can be reduced to equation \eqref{im.red}.
	Substituting  equation \eqref{im} into equation \eqref{gm} gives
	\begin{equation}\notag
		\bm{x}(t) = \bm{\theta}_{\text{L}}\bm{x}(t) +
			\bm{\theta}_{\text{N}} \bm{N}\left(\bm{\eta}_y+\int_{t_1}^{t}\bm{x}(\tau)d\tau\right) + \bm{\beta}
	\end{equation}
	where, differentiating both sides with respect to $t$ and using the chain rule gives equation \eqref{im.red}.
    Then, by combining equations \eqref{gm} and \eqref{im}, the initial value can be obtained as
	\[
	\bm{\eta}_x = \frac{d}{dt}\bm{y}(t) \Big|_{t=1} =
		\bm{\theta}_{\text{L}}\bm{\eta}_y + \bm{\theta}_{\text{N}}\bm{N}(\bm{\eta}_y) + \bm{\beta}.
	\]

	Conversely, integrating both sides of equation \eqref{im.red} with respect $t$, we have, in light of equation \eqref{im},
	\[
	\int_{t_1}^{t} d\bm{x}(\tau) =
		\bm{\theta}_{\text{L}}\int_{t_1}^{t}\bm{x}(\tau)d\tau +
		\bm{\theta}_{\text{N}} \int_{t_1}^{t} \bm{x}(\tau) \frac{d}{d\tau}\bm{N}
		\left(\bm{\eta}_y + \int_{t_1}^{\tau} \bm{x}(s) ds \right) d\tau.
	\]

	Manipulating the second term at the right-hand side through the method of integration by substitution yields
	\begin{align} \notag
	\int_{t_1}^{t}d\bm{x}(\tau)
        = \bm{\theta}_{\text{L}}\int_{t_1}^{t}\bm{x}(\tau)d\tau +
			\bm{\theta}_{\text{N}} \int_{t_1}^{t} \frac{d}{d\tau}\bm{N}\left(\bm{y}(\tau) \right) d \bm{y}(\tau)
		= \bm{\theta}_{\text{L}}\int_{t_1}^{t}\bm{x}(\tau)d\tau +
			\bm{\theta}_{\text{N}} \int_{t_1}^{t} d \bm{N}\left(\bm{y}(\tau) \right)
	\end{align}
	then, using the Newton-Leibniz formula leads to
    \begin{equation}\label{lem.3}
   	\begin{split}
    	\bm{x}(t)
            & = \bm{\theta}_{\text{L}} \left(\bm{y}(t)-\bm{\eta}_y \right) +
    		 \bm{\theta}_{\text{N}}\left( \bm{N}(\bm{y}(t)) -
    		 \bm{N}(\bm{\eta}_y) \right) + \bm{\eta}_x \\
    		 &= \bm{\theta}_{\text{L}} \bm{y}(t) + \bm{\theta}_{\text{N}} \bm{N}(\bm{y}(t)) +
            \bm{\eta}_x - \bm{\theta}_{\text{L}} \bm{\eta}_y - \bm{\theta}_{\text{N}} \bm{N}(\bm{\eta}_y)
    	\end{split}
    \end{equation}
	where, particularly, the left-hand side is
	\[
	\bm{x}(t) = \frac{d}{dt}\left(\bm{\eta}_y + \int_{t_1}^{t}\bm{x}(s)ds \right) =\frac{d}{dt}\bm{y}(t)
	\]
    and substituting the initial condition $\bm{\eta}_x =
	\bm{\theta}_{\text{L}}\bm{\eta}_y + \bm{\theta}_{\text{N}}\bm{N}(\bm{\eta}_y) + \bm{\beta}$ into equation \eqref{lem.3}, we obtain the desired equation \eqref{gm}.
\end{proof}

Theorem \ref{lem1} shows that for a given nonlinear grey equation \eqref{gm}, we can always find an equivalent reduced-order integro-differential equation \eqref{im.red}.
Conversely, equation \eqref{im.red} can be integrated into the unified model \eqref{gm} with the initial value $\bm{\eta}_y$ satisfying $\bm{\eta}_x =
\bm{\theta}_{\text{L}}\bm{\eta}_y + \bm{\theta}_{\text{N}}\bm{N}(\bm{\eta}_y) + \bm{\beta}$.
Note that the constant vector $\bm{\beta}$ provides $d$ degrees of freedom corresponding the initial value $\bm{\eta}_y$.
Let us illustrate it with an example.

\begin{example}
	Supposing a single-output nonlinear grey model (the $\text{GGVM(1,1)}$ model) is
	\begin{align}\label{ex.ver}
		\frac{d}{dt} y(t) = ay(t) + b\left( y(t) \right)^2, ~y(t_1) = \eta_y, ~t\ge 0
	\end{align}
	then the equivalent reduced-order integro-differential equation is
	\begin{align}\label{ex.ver.im}
		\frac{d}{dt}x(t) = ax(t) +2bx(t)\left(\eta_y +\int_{t_1}^{t}x(\tau)d\tau \right),   ~x(t_1)=\eta_x, ~t\ge 0.
	\end{align}
	where the initial value satisfy $\eta_x = a\eta_y + b{\eta_y}^2$. It is easy to verify that their closed form solutions satisfy equation \eqref{im}, that is, the closed form solutions of equations \eqref{ex.ver} and \eqref{ex.ver.im} are
	\[
	y(t) = \left[-\frac{b}{a} + e^{-a(t - t_1)} \cdot \left( \frac{1}{\eta_y} + \frac{b}{a} \right)  \right] ^{-1}
    ~
	\text{and}
	~
	x(t) = \frac{e^{-a (t-t_1)}\left(\frac{b}{a}  + \frac{1}{\eta} \right) }
		{\left[ \frac{b}{a} - e^{-a (t-t_1)} \left( \frac{b}{a}  + \frac{1}{\eta} \right)  \right]^2 },
    ~
    \text{respectively}.
	\]
\end{example}

Specially, if $\bm{\eta}_x=\bm{\eta}_y=\bm{\eta}$, the traditional Cusum operator is the left rectangle approximation of equation \eqref{im}, which has been proved in the literature \citep{wei2020unified,yang_integral_2021}.
In this sense, the integral operator can be regarded as a continuous-time generalisation of the Cusum operator.
In such a case, for a given integro-differential equation \eqref{im.red}, we can always find an equivalent nonlinear grey system model \eqref{gm} with initial value $\bm{y}(t_1) = \bm{\eta}_x$.


\subsection{Integral matching for estimating parameters and initial values simultaneously} \label{4.2}

The reduced-order integro-differential equation concerns the state variable rather than the Cusum state, making the direct modelling possible. By using a state-space form, the reconstructed model can be expressed as
\begin{align}
	\textsf{Observation equation} ~~
        &{\bm{x}}(t_k) = \bm{s}(t_k) + \bm{e}(k), k = 1,2,\cdots,n \label{obs}\\
	\textsf{State equation} ~~
        &\frac{d}{dt}\bm{s}(t) =
	\bm{\theta}_{\text{L}}\bm{s}(t) + \bm{\theta}_{\text{N}}\bm{s}(t)  \frac{d}{dt}\bm{N}\left(\bm{\zeta}+\int_{t_1}^{t}\bm{s}(\tau)d\tau\right),
	~\bm{s}(t_1)=\bm{\eta},~t\ge t_1
	 \label{regm}
\end{align}
where
$\bm{x}(t_k)$ is the observation,
$\bm{e}(k)$ is the measurement noise,
$\bm{s}(t)\in\mathbb{R}^d$ is the state variable,
$\bm{N}(\cdot):\mathbb{R}^d \rightarrow \mathbb{R}^p$ is a $p$-dimensional nonlinear vector function of the integral operator,
$\bm{\theta}_{\text{L}}\in \mathbb{R}^{d \times d}$ is the unknown parameters of linear term,
$\bm{\theta}_{\text{N}}\in \mathbb{R}^{d \times p}$ is unknown structural parameter of nonlinear term,
$\bm{\zeta}\in \mathbb{R}^{d}$ is an unknown constant parameter,
and $\bm{\eta}\in\mathbb{R}^{d}$ is the unknown initial condition.

It is obvious that equation \eqref{regm} is nonlinear in $\bm{\zeta}$, so that this parameter estimation problem falls in the general class of nonlinear least square problem (NLS).
Recall that Theorem \ref{lem1} supports the flexibility selection of $\bm{\zeta}$.
Therefore, for ease of calculation, we take $\bm{\zeta}=\bm{\eta}$, leading to a new form
\begin{equation}\label{restate}
	\frac{d}{dt}\bm{s}(t) =
		\bm{\theta}_{\text{L}}\bm{s}(t) + \bm{\theta}_{\text{N}}\bm{s}(t)  \frac{d}{dt}\bm{N}\left(\bm{\eta}+\int_{t_1}^{t}\bm{s}(\tau)d\tau\right),
		~\bm{s}(t_1)=\bm{\eta},~t\ge t_1.
\end{equation}

In what will follow, the integral matching approach (also called the direct integral method), is utilized to estimate the unknown structural parameters and initial value simultaneously.

Integrating equation \eqref{restate} with respect $t$ over the interval $[t_1,~t_k]$, and combining Theorem \ref{lem1}, we have
\begin{equation}\label{reim}
	\bm{s}(t_k) = \bm{\theta}_{\text{L}}\int_{t_1}^{t_k}\bm{s}(\tau)d\tau +
		\bm{\theta}_{\text{N}}\left[\bm{N}\left(\bm{\eta}+\int_{t_1}^{t_k}\bm{s}(\tau)d\tau \right) - \bm{N}(\bm{\eta}) \right] + \bm{\eta}
\end{equation}
where, by using piecewise linear integration formula, the integral term $\int_{t_1}^{t_k}\bm{s}(\tau)d\tau$ is approximated as
\[
\tilde{\bm{s}}(t_k) \approx \frac{1}{2}\sum_{i=2}^{k}h_i\bm{s}(t_{i-1}) + \frac{1}{2} \sum_{i=2}^{k}h_i\bm{s}(t_{i}).
\]

Due to the state variable $\bm{s}(t)$ is not available, we use the noisy observation $\bm{x}(t)$ instead, then
\begin{equation}\label{int.proxi}
	\tilde{\bm{x}}(t_k) \approx \frac{1}{2}\sum_{i=2}^{k}h_i\bm{x}(t_{i-1}) + \frac{1}{2} \sum_{i=2}^{k}h_i\bm{x}(t_{i}).
\end{equation}

Correspondingly, the pseudo-nonlinear regression problem becomes
\begin{equation}\label{im.dis}
	\bm{x}(t_k) = \bm{\theta}_{\text{L}}\tilde{\bm{x}}(t_k) +
			\bm{\theta}_{\text{N}}\left[ \bm{N}(\bm{\eta}+\tilde{\bm{x}}(t_k))-\bm{N}(\bm{\eta})\right]
			+ \bm{\eta} + \bm{e}(t_k)
\end{equation}
where $\bm{e}(t_k)$ is the sum of discretization error and measurement noise.
Then, we utilize the change of basis to show how to derive this nonlienar least squares problem to a linear version.

\begin{lemma}\label{lem2}	
	Suppose that $\bm{N}(\cdot)$ consists of polynomial terms. Depending on whether there exist interaction terms or not we get the following alternatives. 
	
	\begin{enumerate}
		\item 
		If $d=1$ and ${N}(x)=\left[x^2,~ x^3,~ \cdots,~ x^{p+1}\right]^\top$, then 
		\begin{align}\label{im2}
			{\theta}_{\text{L}} \tilde{{x}}(t_k) + 
			\bm{\theta}_{\text{N}} \left[ \bm{N}(\bm{\eta}+\tilde{\bm{x}}(t_k))-\bm{N}(\bm{\eta})\right] + 
			{\eta} 
			= 
			\left( {\theta}_{\text{L}} + {\theta}_{\text{L}}\bm{\phi} \right) \tilde{{x}}(t_k) +
			\left( \bm{\theta}_{\text{N}}\bm{\varphi} \right) 
			\bm{N}(\tilde{\bm{x}}(t_k)) +
			\eta
		\end{align}
		where 
		\[
		\bm{\phi} = \begin{bmatrix}
			\binom{2}{1} \eta^{1} \\
			\vdots \\
			\binom{p+1}{1}\eta^{p}
		\end{bmatrix},
		~
		\bm{\varphi} = \begin{bmatrix}
			\binom{2}{2} \eta^{0} & & \\
			\vdots & \ddots & \\
			\binom{p+1}{2} \eta^{p-1} & \cdots & \binom{p+1}{p+1} \eta^{0}
		\end{bmatrix}.
	\]
	
		\item 
		If $d\geq 2$ and $\bm{N}(\bm{x})=\left[x_1^2,~ x_1x_2,~ \cdots,x_1x_d,~x_2^2,\cdots,~ x_d^2\right]^\top$, then
		\begin{align}\label{im3}
			\bm{\theta}_{\text{L}} \tilde{\bm{x}}(t_k) + 
			\bm{\theta}_{\text{N}} \left[ \bm{N}(\bm{\eta}+\tilde{\bm{x}}(t_k))-\bm{N}(\bm{\eta})\right] + 
			\bm{\eta} = 
			\left( \bm{\theta}_{\text{L}}\bm{I}_d + \bm{\theta}_{\text{N}}\bm{\psi} \right) \tilde{\bm{x}}(t_k) +
			\bm{\theta}_{\text{N}} \bm{N}(\tilde{\bm{x}}(t_k)) +
			\bm{\eta}
		\end{align}
		where $\bm{\psi}=\begin{bmatrix}
			\psi_d ^\top & \psi_{d-1} ^\top & \cdots & \psi_1 ^\top
		\end{bmatrix}^\top$
		for
		\[
		\psi_{d-i} =  
		\underset{i \text{ columns}}{\underbrace{
				\left[ 
				\begin{array}{cccc}
					0 		& \cdots & 0		\\
					0 		& \cdots & 0		\\
					\vdots  & \ddots & \vdots		\\
					0		& \cdots & 0		\\
				\end{array}
				\right.
		}}
		\;
		\underset{d-i \text{ columns}}{\underbrace{
				\left.
				\begin{array}{cccc}
					2\eta_{i+1}		&				&		&	0	\\
					\eta_{i+2}		&\eta_{i+1}		&		&		\\
					\vdots 			&				& \ddots&		\\
					\eta_d 			&0				&		&	\eta_{i+1}	\\
				\end{array}
				\right]
		}},
		~
		i = 0,~1,\cdots,~d-1.
		\]
	\end{enumerate}
\end{lemma}

Due to space limitations, we provide the technical proof in Appendix \ref{app1}.
Lemma \ref{lem2} shows that by manipulating the nonlinear function vector, $\bm{\eta}$ can be separated, thereby transferring the nonlinear least squares to a simple linear one.
Lemma \ref{lem2} holds for most of the nonlinear grey models, although only the case of quadratic nonlinearities is presented here.
For higher order polynomials and other forms such as trigonometric, Lemma \ref{lem2} provides a reference.

For simplicity of presentation, the transformed parameters at the right-hand side of equations \eqref{im2} and \eqref{im3} are denoted as
$\left[ \bm{\vartheta}_{\text{L}} ~ \bm{\vartheta}_{\text{N}} ~ \bm{\eta}\right]$.
Sequentially, applying Lemma \ref{lem2} to equation \eqref{im.dis}, we can obtain a pseudo-linear regression
\begin{equation}\label{im.tra}
	\bm{x}(t_k) = \bm{\vartheta}_{\text{L}}\tilde{\bm{x}}(t_k) +
		\bm{\vartheta}_{\text{N}}\bm{N}\left( \tilde{\bm{x}}(t_k)\right) + \bm{\eta} + \bm{e}(t_k)
\end{equation}
then, similar to the manipulation in Section \ref{sec2}, one has matrix form
\begin{equation}
	\bm{X} = \Omega(\bm{x}){\Pi} + {\bm{E}}
\end{equation}
where
\[
{\Pi} =
	\begingroup
	\renewcommand*{\arraystretch}{1.2}
	\begin{bmatrix}
		\bm{\vartheta}^\top_{\text{L}} \\
		\bm{\vartheta}^\top_{\text{N}} \\
		\bm{\eta}^\top
	\end{bmatrix}
	\endgroup,
~
\Omega(\bm{x}) =
\begin{bmatrix}
	\tilde{\bm{x}}^\top(t_2)	 &  \bm{N}^\top \left(\tilde{\bm{x}}^\top(t_2)\right)  & 1 \\
	\tilde{\bm{x}}^\top(t_3)	 &  \bm{N}^\top \left(\tilde{\bm{x}}^\top(t_3)\right)  & 1 \\
	\vdots & \vdots & \vdots \\
	\tilde{\bm{x}}^\top(t_n)	 &  \bm{N}^\top \left(\tilde{\bm{x}}^\top(t_n)\right)  & 1 \\
\end{bmatrix},
~
\bm{E} =
\begin{bmatrix}
	\bm{e}^\top(t_2) \\
	\bm{e}^\top(t_3) \\
	\vdots \\
	\bm{e}^\top(t_n)
\end{bmatrix}.
\]

Minimizing the least-squares objective function $\bm{\mathcal{L}}({\Pi}) = \| \bm{X}-\Omega(\bm{x}){\Pi} \|^2_\mathsf{F}$ gives the simultaneous estimates of structural parameters and initial value:
\begin{equation}
	\hat{\Pi} = \left( \Omega^\top(\bm{x}) \Omega(\bm{x}) \right)^{-1} \Omega^\top(\bm{x})\bm{X}.
\end{equation}

Then, according to equation \eqref{im2} or \eqref{im3},  we can obtain the estimates $\begin{bmatrix}
	\hat{\bm{\theta}}_{\text{L}} ~ 
	\hat{\bm{\theta}}_{\text{N}} ~
	\hat{\bm{\eta}}
\end{bmatrix}$ by inversely solving algebraic equations.
Finally, substituting the estimates into the integro-differential equation model \eqref{regm} gives the time response function
\[
\hat{\bm{x}}(t) = \bm{f}\left(\hat{\bm{\eta}},\hat{\bm{\theta}}_{\text{L}},\hat{\bm{\theta}}_{\text{N}};t  \right) 
\]
and then by substituting time points $\left\lbrace t_k \right\rbrace ^{n+r}_{k=1}$, we can obtain the fitting and forecasting values of the original time series $\left\lbrace \hat{\bm{x}}(t_k) \right\rbrace^{n+r}_{k = 1} $ directly.

It should be noticed that if $\bm{N}(\cdot)$ includes a power term, then equations \eqref{im2} and \eqref{im3} do not hold any more. In such a case, an alternative strategy to trade off the model accuracy and computation complexity, is using $\bm{x}(t_1)$ to replace $\bm{\eta}$ when discretizing $\bm{N}\left(\bm{\eta}+\int_{t_1}^{t_k}\bm{s}(\tau)d\tau \right)$ in equation \eqref{reim}, leading to
\[
\bm{x}(t_k) = \bm{\theta}_{\text{L}}\tilde{\bm{x}}(t_k) +
	\bm{\theta}_{\text{N}}\left[ \bm{N}(\bm{x}(t_1)+\tilde{\bm{x}}(t_k))-\bm{N}(\bm{x}(t_1))\right]
	+ \bm{\eta} + \bm{e}(t_k)
\]
and, by performing the linear least squares criterion, it is straightforward to obtained the structural parameters and initial value estimates; see \cite{yang_integral_2021} for a particular example.

\subsection{Comparison between nonlinear grey models and integro-differential equation models}

In order to analyse the relationship between nonlinear grey system model and its integro-differential equation-based reconstruction, we summarise the modelling procedures of both models in Figure \ref{remodel}.

\begin{figure}[htp]
    \centering
    \includegraphics[width=.95\linewidth]{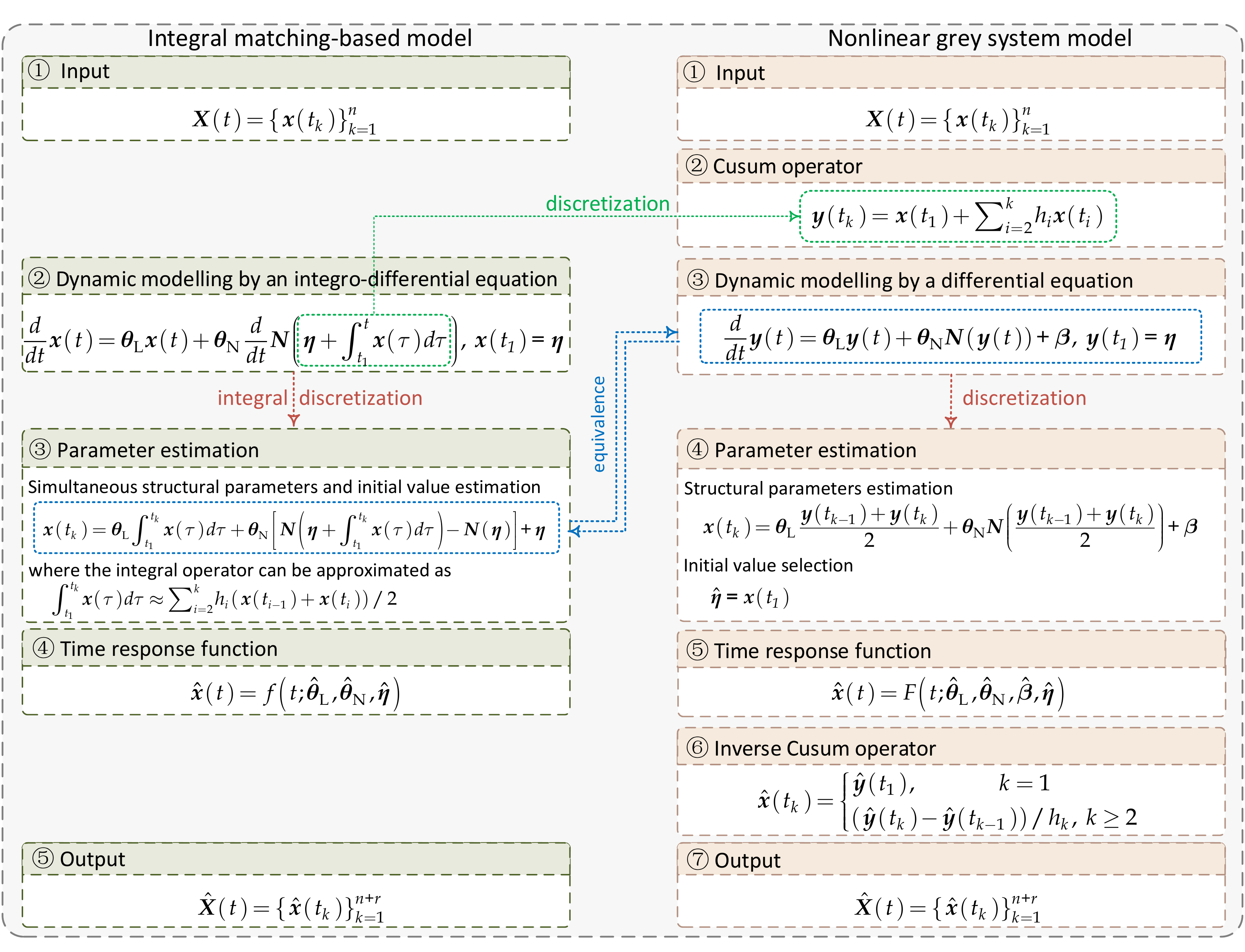}
    \caption{Modelling procedure comparison between nonlinear grey models and integro-differential equation models.}
    \label{remodel}
\end{figure}

It can be seen in Figure \ref{remodel} that the integral operator (or its numerical discretization-based Cusum operator) bridges two models, that is, both use the integral operator (Cusum operator) but in different manner. In particular:

\begin{enumerate}
	\item[(1)] The Cusum operator is the piecewise left-constant quadrature formula of the integral operator. Besides, the existing variations of Cusum operators, such as the fractional cumulative sum operator \citep{wu_forecasting_2019} and new information priority accumulation \cite{wu2018grey}, can be regarded as different numerical discretization forms of the integral operator.
	\item[(2)]The integro-differential equation can be converted to the nonlinear grey model via a simple integral manipulation, indicating that the nature of the nonlinear grey model is an integro-differential equation model.
	\item[(3)] It is clear that the modelling process of nonlinear grey models can be viewed as the reordered version of the integral matching-based integro-differential equation models.
	However, the nonlinear grey modelling requires the inverse cumulative sum step, likely compromising the modelling accuracy.
	\item[(4)] By applying the integral matching approach to integro-differential equation models, the structural parameters and initial value are obtained simultaneously, whereas the nonlinear grey models need an extra initial value selection strategy (see details in section \ref{sec2}).
\end{enumerate}


\section{Simulations}\label{sec5}
In this section we conduct large-scale Monte Carlo experiments to evaluate the finite sample performance of the nonlinear grey system model and it reconstructed one.

Considering the reduced integro-differential equation of the nonlinear grey Verhulst model with $d=1$, the time series are generated from the state space equation:
\[
\begin{split}
	\textsf{State euqation}~~ &\frac{d}{dt}x(t) = ax(t) + bx(t)\left(\eta +\int_{t_1}^{t}x(\tau)d\tau \right),   ~x(t_1)=\eta, ~t_1 \ge 0 \\
	\textsf{Observation equation}~~ &x(t_k) = s(t_k) + e(k), ~e(k) \sim \mathcal{N}(0,\sigma^2)
\end{split}
\]
where the structural parameters are $a = 1.2, ~b = -1$ and the initial value is $\eta = 0.4$.
According to Theorem \ref{rm1}, it is easy to show the corresponding grey Verhulst model is
\[
\frac{d}{dt} y(t) = ay(t) + \frac{b}{2} \left( y(t) \right)^2, ~y(t_1) = \eta, ~t\ge 0.
\]

\subsection{Experimental fashions and performance criterion}

Since the finite sample performance is affected by the amount of data and measurement noise, they are set to different values.
On the one hand, by varying the data scale, we sample data at every time interval of $h$ in the range of $t \in [0, T]$, thereby generating $n = \left[ \frac{T}{h} \right]  + 1$ samples.
On the other hand, we set different magnitudes of Gaussian noise which are controlled by the noise level:
\begin{equation}\label{eq4.2}
	\textsf{Noise Level} ~(\%)= \frac{\mathsf{var}(\mathrm{Noise})}{\mathsf{var}(\mathrm{Signal})} \times 100\% =
	\frac{\sigma^2}{\mathsf{var}[{s}]} \times 100\%, ~i = 1,2, \cdots, d.
\end{equation}

Specifically, we consider 2 experimental set-ups. Let $T = 4$. In the first set-up, we set the noise level to $10\%$ but change the time interval $h=[0.40, ~0.20, ~0.08, ~0.04]$, thereby generating $n = [11, ~21, ~51, ~101]$ samples.
In the second set-up, we fix the data scale $n=501$ ($h=0.01$) but change the noise level $[10\%, ~15\%, ~20\%, ~25\%]$.
In each case 500 Monte Carlo realizations are repeated.

For ease of comparison, the parameter and initial value estimation performance are depicted by the violin plot synergistically combining the boxplots and the density trace, where the boxplots shows centre, spread, asymmetry and outliers, and the density trace shows the distributional characteristics of batches of data \cite{2001Statistical}.
In addition, the fitting performance is measured by the mean absolute percentage error criteria
\[
\textsf{RMSE}[{x}] = \sqrt{ \frac{1}{n-1} \sum_{k=1}^{n} \left(\hat{x}(t_k) - x(t_k)\right)^2 }.
\]

In the following, the universal modelling framework of the nonlinear grey system is abbreviated as nonlinear grey modelling, and the reconstructed modelling process is abbreviated as integral matching.
Note, the fixed point strategy is employed to select initial conditions for solving the original nonlinear grey model.


\subsection{Performance evaluation under varying data scales}

The estimation of structural parameters and initial values, as well as the fitting errors obtained from integral matching and nonlinear grey modelling are summarised in Figure \ref{ver_nlen}.

\begin{figure}[h]
	\centering
	\includegraphics[width=1\linewidth]{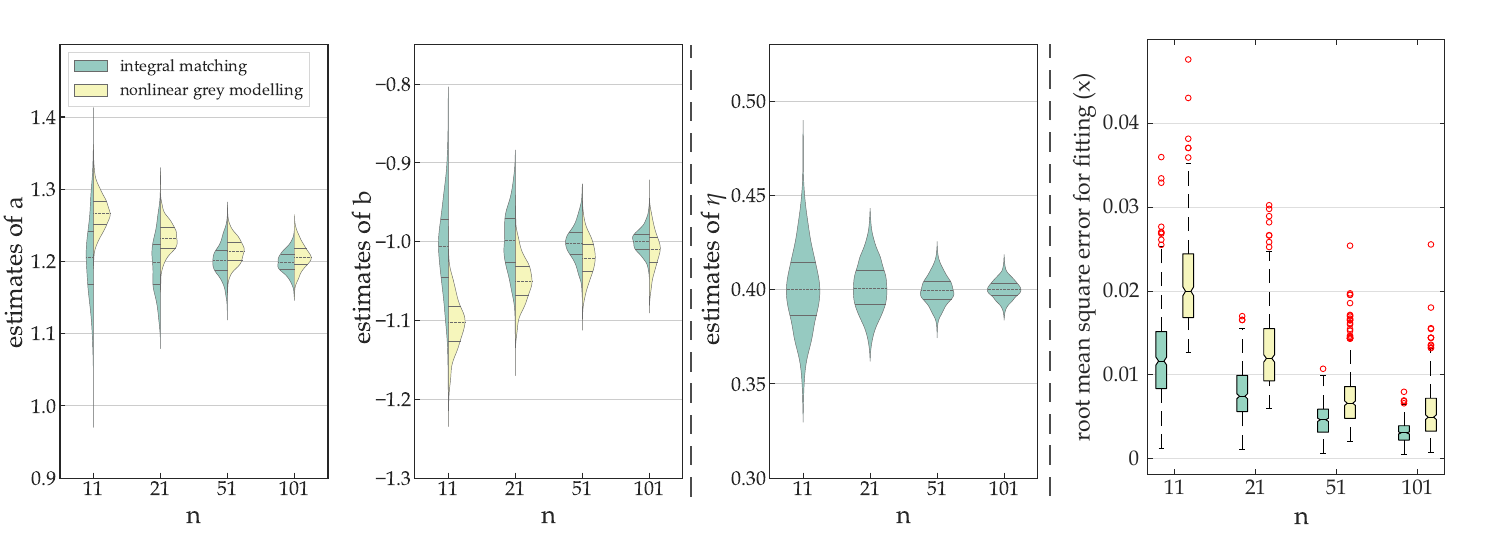}
	\caption{Violin plots of estimated structural parameters, initial value and boxplot of fitting error.
		The true parameters are $a=1.2,~b=-1$ and $\eta=0.4$,
		and the nonlinear grey modelling uses $x(t_1)$ as initial value.
	}
	\label{ver_nlen}
\end{figure}

Overall, based on the results, we can see that integral matching outperforms nonlinear grey modelling, including parameter estimation accuracy and fitting accuracy.
In terms of parameter estimation, with the increasing of sample size, the estimates of both approaches go to the true $a = 1.2, ~b = -1.0$, suggesting the asymptotic property of integral matching and nonlinear grey modelling.
Apparently, $\hat{a}$ and $\hat{b}$ of integral matching tend to have symmetric violin plots, both centring around the true values $ 1.2$ and $-1.0$ in all cases, indicating the likely unbiased estimation of integral matching; whereas the nonlinear grey modelling generates biased estimators even in the case of large data size ($n = 101$) combination.
Meanwhile, when encountered samples are small sized, $n = 11$, the variance (measured by the violin shape) of nonlinear grey modelling is smaller than that of integral matching, suggesting the robustness of the former is slightly better than that of the later.
As $n$ grows, the variance of both methods decreases, but the decrease of nonlinear grey modelling has a smaller rate than integral matching.
As for initial value estimation, the estimates approaches the true value $\eta=0.4$ as the increase of data scale, validating the efficiency of the proposed estimation method.
From the perspective of fitting performance, as the data size increases, the fitting performance of both models ameliorates.
Meanwhile, integral matching has lower RMSEs and less outliers  in all cases, indicating higher fitting accuracy than nonlinear grey modelling.


\subsection{Performance evaluation under varying noise levels}

The above results demonstrate that the large data scale performs best, thus to probe the effect of noise level, we carry on the experiment in the case of $n=101$.
Figure \ref{ver_nois} summarizes the distributions of estimated structural parameters, initial values and fitting errors of integral matching versus nonlinear grey modelling.

\begin{figure}[h]
	\centering
	\includegraphics[width=1\linewidth]{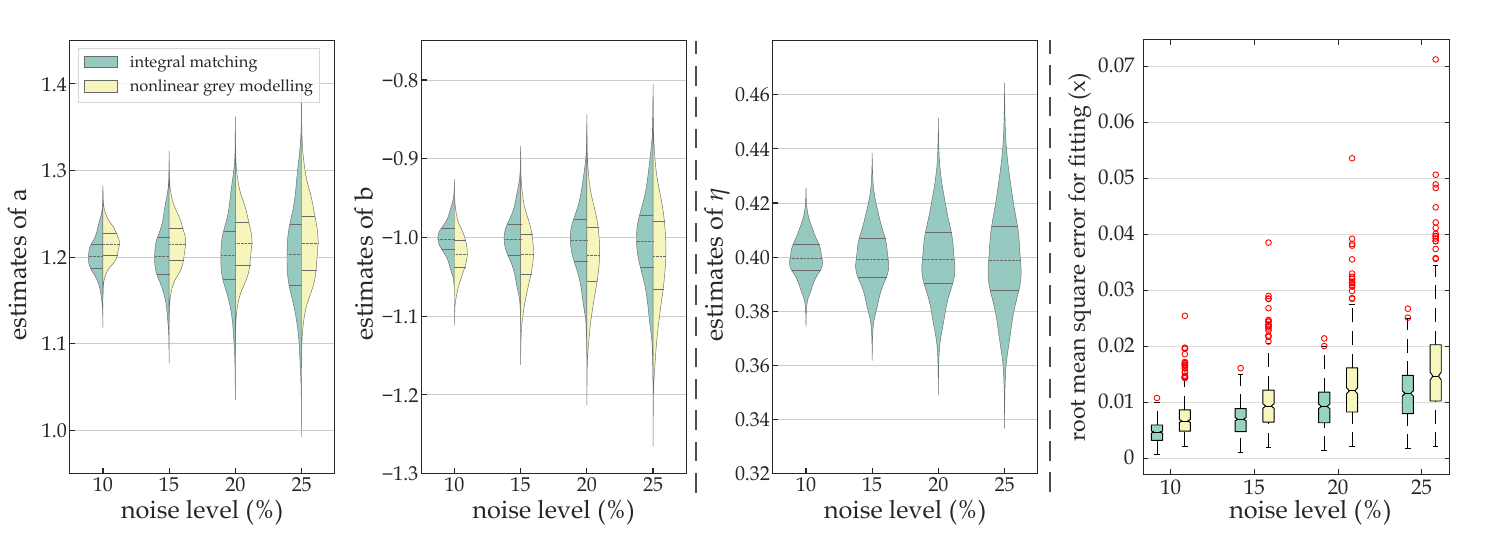}
	\caption{Violin plots of estimated structural parameters, initial value and boxplot of fitting error.
		The true parameters are $a=1.2,~b=-1$ and $\eta=0.4$, and nonlinear grey modelling uses $x(t_1)$ as initial value.}
	\label{ver_nois}
\end{figure}

As a whole, both integral matching and nonlinear grey modelling perform satisfactory robustness to noise in all noise magnitudes, although the noise impacts the parameter estimation and fitting accuracy.
To be specific, from the perspective of parameter estimation accuracy, there only exist slight changes of the medians for both method with the increase of noise level, whereas the distribution of the estimates becomes scattered, indicating the high noise level leads to more outliers.
Meanwhile, the two modelling methods depict similar shapes of violin plots, indicating comparable performance of robustness to noise without a substantial difference.
Similarly, as for initial value distribution, the increasing noise magnitude brings tiny effect on the median of violin plots for integral matching, while the variance increase obviously. 
In addition, in the case of large noise level (25\%), integral matching still remain higher fitting accuracy, suggesting the efficiency of this kind of modelling method.
However, nonlinear grey have some large outliers for the high noise magnitude.


\subsection{Short discussion}

In the single-output scenario, integral matching and nonlinear grey modelling shows several similar results, mainly reflecting in robustness to noise and parameter estimation accuracy for large data size.
This is, however, not a general conclusion.
In the high-dimensional systems, nonlinear grey modelling performs distinctly worse than integral matching, in terms of parameter estimation accuracy, robustness to noise and fitting accuracy.
Take an multi-output system ($d=2$) as an example, we summarize the modelling results in Appendix \ref{append1}.
Particularly, we find that using the first point of the noisy data for solving the original nonlinear grey model leads to blow-up solutions, the reason of which is the loss of accuracy of the numerical method \cite{ascher1998computer}, and thus it cannot deal with measurement noise.
For illustrating experimental results intuitively, we use the true values of the initial condition instead to solve the 2-dimensional nonlinear grey model.

\section{Real-world application} \label{sec6}

Water is vital for human health, industry, agriculture and energy production.
Yet with the rapid expansion of urbanization and growth of population, the overall demand for water and the quantity of wastewater produced are both continuously increasing worldwide, which brings formidable threats to the world's water systems \citep{srinivasan2012nature}.
The Yangtze River Delta (YRD) in China, consisting of Shanghai, Jiangsu, Zhejiang, and Anhui Provinces, is one of the key regions for water shortage and pollution due to dense population, fast economic development and high urbanization.
Figure \ref{water} shows that the municipal sewage discharge and the total amount water use of the YRD are higher than other regions.

Over the past decades, urban sewage has begun to be comprehensively utilized as a potential water resource, which may alleviate the pressure of the ecological environment and speed the development of the regional economy \citep{water20172017}.
Wastewater recovery becomes an efficient measure for water-saving economic patterns.
A scientific prediction of municipal sewage is the basis of urban drainage system design, operation and management.
Therefore, in this paper, by using the proposed reconstructed nonlinear grey models, we predict the municipal sewage and the total amount of water use to probe the water resources development of the YRD.
The data in the period from 2004 to 2018 is collected from \textsf{China Statistic Yearbook on Environment}.

\begin{figure}[h]
	\centering
	\includegraphics[width=1\linewidth]{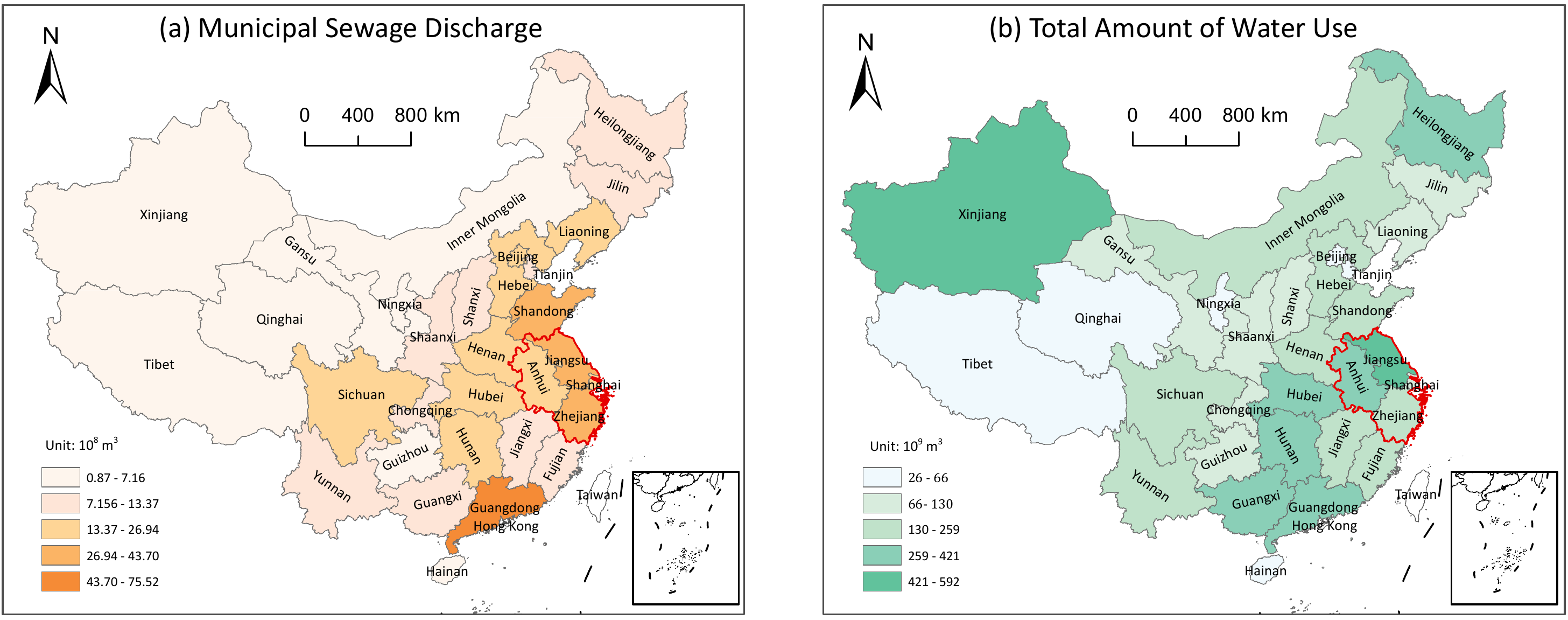}
	\caption{Municipal sewage discharge and total amount of water use of provinces (cities and districts) in 2018.}
	\label{water}
\end{figure}

\subsection{Results of integro-differential equation models}

In order to evaluate the performance of the proposed models, we divide the time series into two parts: 2004-2014 data used to train the models and 2015-2018 data used for assessing the prediction performance.
Then, we proceed three-step ahead forecasting for the municipal sewage discharge and water consumption from 2019-2021 to analyse the post-2018 behaviour and support our suggestions.

The models (\romannumeral1)-(\romannumeral3) in Tables \ref{result_sewage} and \ref{result_consumption} are three reconstructed integro-differential equation models, almost covering the existing single-output nonlinear grey models (shown as Table \ref{review}).
And the three models are correspondingly referred to as IGVM, INGM, and INGBM, respectively.
INGM and INGBM model have an unknown power term $\gamma$, in the following, we utilize line search \citep{nocedal2006numerical} to select an appropriate one.
Assuming $\gamma \in [a, b]$ on the basis of prior knowledge, there exists $N = \frac{b-a}{\lambda} + 1$ possible choice for $\gamma$, where $\lambda$ is the search step length and, here, set $a=0,~b=2$ and $\lambda=0.01$.
For each possibility, we perform the parameter estimation procedures for the real data and choose the best one (measured by forecasting error).
Meanwhile, we use the alternative parameter estimation method for power model families mentioned in subsection \ref{4.2} to address the INGM and INGBM models.

The results obtained from the aforementioned methods are summarized in Tables \ref{result_sewage} and \ref{result_consumption}.
From Table \ref{result_sewage}, we can see that the three integro-differential equation models have similar performance for fitting the municipal sewage discharge data with 2.57\%, 2.58\% and 2.62\% MAPE$_\mathsf{train}$s, respectively.
While the forecasting errors of INGM and INGBM are 0.56\% and 0.87\%, which is lower than the IGVM (MAPE$_\mathsf{test}$=2.63\%), indicating that models INGM and INGBM are superior to IGVM model.
Besides, Table \ref{result_consumption} shows that for the total amount of water use scenario, similarly, there are only small differences in the fitting error among the models (\romannumeral1)-(\romannumeral3).
IGVM has the smallest MAPE$_\mathsf{train}$ for 0.75\%, whereas the testing error is the largest one (4.03\%) among nonlinear grey models.
Apparently, the INGBM shows better forecasting performance than the IGVM and INGM models with MAPE$_\mathsf{train}$ as low as 1.29\%.

To sum up, INGBM performances best for both cases among models (\romannumeral1)-(\romannumeral3).
Therefore, it is probably the best one to use in this application.
In what will follow, we detail the modelling procedures of INGBM model.

Fitting this model on the time series gives the estimated parameters are $a = 0.014,~b = 0.014,~\gamma=1$ and $\eta=77.33$.
The corresponding expression of the INGBM model is
$\frac{d}{dt}x(t) = 0.028x(t)$ with $x(t_1) = 77.33$.

Similarly, for the case of total amount of water use, the estimates are  $a = -0.06, ~b = 2.89,~\gamma=0.63,~\eta=1001.65$, and model expression is
$
\frac{d}{dt}x(t) = -0.06x(t)+2.89x(t)\left( 1001.65 + \int_{t_1}^{t}x(t) \right)^{-0.37}
$ with $x(t_1) = 1001.65$.

\begin{table}[h]
	\begin{threeparttable}
	\centering
	\caption{Training and testing results for municipal sewage discharge obztained from reconstructed nonlinear grey models and comparison methods.}
	\scriptsize
	\newcommand{\tabincell}[2]{\begin{tabular}{@{}#1@{}}#2\end{tabular}}
	\begin{tabular}{llllllllllllll}
		\toprule
		\multirow{2}{*}{Year} & \multirow{2}{*}{ \tabincell{l}{True values \\ ($10^8$ m$^3$)}} & IGVM$^\mathrm{\romannumeral1}$   &       & INGM$^\mathrm{\romannumeral2}$   &       & INGBM$^\mathrm{\romannumeral3}$  &      & ARIMA$^\mathrm{\romannumeral4}$    &       & SVR$^\mathrm{\romannumeral5}$    &      & NNAR$^\mathrm{\romannumeral6}$  &      \\
		\cmidrule(lr){3-4}\cmidrule(lr){5-6}\cmidrule(lr){7-8}\cmidrule(lr){9-10}\cmidrule(lr){11-12}\cmidrule(lr){13-14}
		&                              & Values & APE  & Values & APE & Values & APE & Values & APE  & Values & APE & Values & APE \\
		\midrule
		2004                  & 83.00                                 & 78.85        & 5.00      & 77.92        & 6.12      & 77.33        & 6.84       & --           & --         & --          & --        & --          & --        \\
		2005                  & 85.58                                 & 80.46        & 5.98      & 79.83        & 6.72      & 79.54        & 7.06       & 84.93        & 0.76       & --          & --        & --          & --        \\
		2006                  & 77.89                                 & 82.22        & 5.56      & 81.90        & 5.14      & 81.81        & 5.03       & 87.50        & 12.34      & 90.26       & 15.89     & 87.28       & 12.06     \\
		2007                  & 80.45                                 & 84.13        & 4.59      & 84.09        & 4.53      & 84.15        & 4.60       & 79.81        & 0.79       & 80.22       & 0.28      & 80.86       & 0.51      \\
		2008                  & 87.27                                 & 86.23        & 1.19      & 86.40        & 0.99      & 86.55        & 0.82       & 82.37        & 5.61       & 87.05       & 0.24      & 82.07       & 5.96      \\
		2009                  & 88.11                                 & 88.51        & 0.45      & 88.83        & 0.82      & 89.03        & 1.04       & 89.19        & 1.22       & 88.40       & 0.32      & 88.42       & 0.35      \\
		2010                  & 92.53                                 & 90.99        & 1.66      & 91.38        & 1.25      & 91.57        & 1.04       & 90.04        & 2.70       & 92.35       & 0.20      & 90.09       & 2.64      \\
		2011                  & 95.08                                 & 93.70        & 1.44      & 94.04        & 1.09      & 94.19        & 0.93       & 94.46        & 0.65       & 94.30       & 0.82      & 94.10       & 1.03      \\
		2012                  & 97.88                                 & 96.65        & 1.25      & 96.81        & 1.09      & 96.88        & 1.02       & 97.00        & 0.89       & 97.33       & 0.55      & 96.93       & 0.96      \\
		2013                  & 99.82                                 & 99.87        & 0.05      & 99.71        & 0.11      & 99.65        & 0.17       & 99.80        & 0.02       & 99.60       & 0.22      & 99.74       & 0.08      \\
		2014                  & 102.24                                & 103.38       & 1.11      & 102.72       & 0.47      & 102.50       & 0.25       & 101.74       & 0.49       & 101.62      & 0.60      & 101.83      & 0.40      \\
		\multicolumn{2}{l}{MAPE$_\mathsf{train}$  (\%)}                              &              & 2.57      &              & 2.58      &              & 2.62       &              & 2.55       &             & 2.12      &             & 2.67      \\
		2015                  & 106.27                                & 107.21       & 0.88      & 105.86       & 0.39      & 105.43       & 0.80       & 104.17       & 1.98       & 103.731     & 2.39      & 104.23      & 1.92      \\
		2016                  & 110.02                                & 111.40       & 1.25      & 109.13       & 0.81      & 108.44       & 1.44       & 106.09       & 3.57       & 105.41      & 4.19      & 106.31      & 3.37      \\
		2017                  & 111.40                                & 115.98       & 4.11      & 112.53       & 1.01      & 111.54       & 0.12       & 108.01       & 3.04       & 107.06      & 3.90      & 108.41      & 2.69      \\
		2018                  & 116.05                                & 121.00       & 4.27      & 116.07       & 0.02      & 114.73       & 1.14       & 109.94       & 5.26       & 108.76      & 6.28      & 110.54      & 4.75      \\
		\multicolumn{2}{l}{MAPE$_\mathsf{test}$  (\%)}                                         &              & 2.63      &              & 0.56      &              & 0.87       &              & 3.47       &             & 4.79      &             & 3.60     \\
		\bottomrule
	\end{tabular}
\label{result_sewage}
			\begin{tablenotes}
				\footnotesize
				\item []				
				Note that the model structures are
				(\romannumeral1) $\frac{d}{dt}x(t) = a x(t) + b x(t)y(t)$;
				(\romannumeral2) $\frac{d}{dt}x(t) = b x(t)y^{\gamma-1}(t) $;
				(\romannumeral3) $\frac{d}{dt}x(t) = a x(t) + b x(t)y^{\gamma-1}(t)$;
				(\romannumeral4) the model order is arima(0,1,0);
				(\romannumeral5) the embedding dimension is 2 and the kernel type is polynomial;
				(\romannumeral6) the model type is feed-forward neural network with 2 lagged inputs, 3 neurons in the only hidden layer. 
			\end{tablenotes}
	\end{threeparttable}
\end{table}

\begin{table}[h]
	\begin{threeparttable}
	\centering
	\caption{Training and testing results for total amount of water use obtained reconstructed nonlinear grey models and comparison methods.}
	\scriptsize
	\newcommand{\tabincell}[2]{\begin{tabular}{@{}#1@{}}#2\end{tabular}}
	\begin{tabular}{llllllllllllll}
		\toprule
		\multirow{2}{*}{Year} & \multirow{2}{*}{ \tabincell{l}{True values \\ ($10^9$ m$^3$)}} & IGVM$^\mathrm{\romannumeral1}$   &       & INGM$^\mathrm{\romannumeral2}$   &       & INGBM$^\mathrm{\romannumeral3}$  &      & ARIMA$^\mathrm{\romannumeral4}$    &       & SVR$^\mathrm{\romannumeral5}$    &      & NNAR$^\mathrm{\romannumeral6}$  &      \\
		\cmidrule(lr){3-4}\cmidrule(lr){5-6}\cmidrule(lr){7-8}\cmidrule(lr){9-10}\cmidrule(lr){11-12}\cmidrule(lr){13-14}
		&                              & Values & APE  & Values & APE & Values & APE & Values & APE  & Values & APE & Values & APE \\
		\midrule
		2004                  & 1061.25                               & 1038.08      & 2.18      & 1040.04      & 2.00      & 1001.65       & 5.62      & --            & --        & --           & --       & --           & --       \\
		2005                  & 1058.94                               & 1071.36      & 1.17      & 1082.30      & 2.21      & 1064.60       & 0.53      & 1073.66       & 1.17      & --           & --       & --           & --       \\
		2006                  & 1115.09                               & 1101.01      & 1.26      & 1107.61      & 0.67      & 1106.32       & 0.79      & 1012.89       & 4.35      & 1114.31      & 0.07     & 1079.13      & 3.23     \\
		2007                  & 1121.56                               & 1126.53      & 0.44      & 1125.75      & 0.37      & 1135.37       & 1.23      & 1118.72       & 0.33      & 1134.27      & 1.13     & 1107.78      & 1.23     \\
		2008                  & 1161.07                               & 1147.49      & 1.17      & 1139.90      & 1.82      & 1155.52       & 0.48      & 1092.15       & 2.62      & 1145.29      & 1.36     & 1135.32      & 2.22     \\
		2009                  & 1164.05                               & 1163.51      & 0.05      & 1151.51      & 1.08      & 1168.91       & 0.42      & 1168.19       & 0.61      & 1157.22      & 0.59     & 1158.96      & 0.44     \\
		2010                  & 1174.64                               & 1174.31      & 0.03      & 1161.36      & 1.13      & 1176.93       & 0.20      & 1163.56       & 0.04      & 1178.68      & 0.34     & 1177.50      & 0.24     \\
		2011                  & 1173.80                                & 1179.69      & 0.50      & 1169.92      & 0.33      & 1180.56       & 0.58      & 1185.58       & 0.93      & 1178.47      & 0.40     & 1184.38      & 0.90     \\
		2012                  & 1155.60                                & 1179.54      & 2.07      & 1177.47      & 1.89      & 1180.54       & 2.16      & 1202.10       & 2.41      & 1190.61      & 3.03     & 1188.48      & 2.85     \\
		2013                  & 1194.20                                & 1173.88      & 1.70      & 1184.25      & 0.83      & 1177.45       & 1.40      & 1127.10       & 2.47      & 1204.65      & 0.88     & 1192.83      & 0.11     \\
		2014                  & 1162.20                                & 1162.79      & 0.05      & 1190.38      & 2.43      & 1171.73       & 0.82      & 1236.30       & 3.52      & 1180.00      & 1.53     & 1191.07      & 2.48     \\
		\multicolumn{2}{l}{MAPE$_\mathsf{train}$  (\%)}                                       &              & 0.75      &              & 1.24      &               & 0.91      &               & 1.62      &              & 1.04     &              & 1.52     \\
		2015                  & 1153.10                                & 1146.50      & 0.57      & 1195.99      & 3.72      & 1163.79       & 0.93      & 1172.30       & 0.87      & 1239.50      & 7.49     & 1237.28      & 7.30     \\
		2016                  & 1154.00                                  & 1125.29      & 2.49      & 1201.16      & 4.09      & 1153.93       & 0.01      & 1182.39       & 2.54      & 1238.74      & 7.34     & 1216.65      & 5.43     \\
		2017                  & 1165.90                                & 1099.54      & 5.69      & 1205.94      & 3.43      & 1142.42       & 2.01      & 1192.49       & 3.34      & 1241.89      & 6.52     & 1259.58      & 8.03     \\
		2018                  & 1155.00                                  & 1069.69      & 7.39      & 1210.40      & 4.80      & 1129.51       & 2.21      & 1202.58       & 3.15      & 1239.07      & 7.28     & 1251.82      & 8.38     \\
		\multicolumn{2}{l}{MAPE$_\mathsf{test}$  (\%)}                                      &              & 4.03      &              & 4.01      &               & 1.29      &               & 2.47      &              & 7.05     &              & 7.29    \\
		\bottomrule
	\end{tabular}
\label{result_consumption}
		\begin{tablenotes}
			\item []
			Note that all the models in this table share same structures with those in Table \ref{result_sewage}.
		\end{tablenotes}
\end{threeparttable}
\end{table}

\begin{figure}[h]
	\centering
	\subfigure{
		\begin{minipage}[t]{0.5\linewidth}
			\centering
			\includegraphics[width=0.85\linewidth]{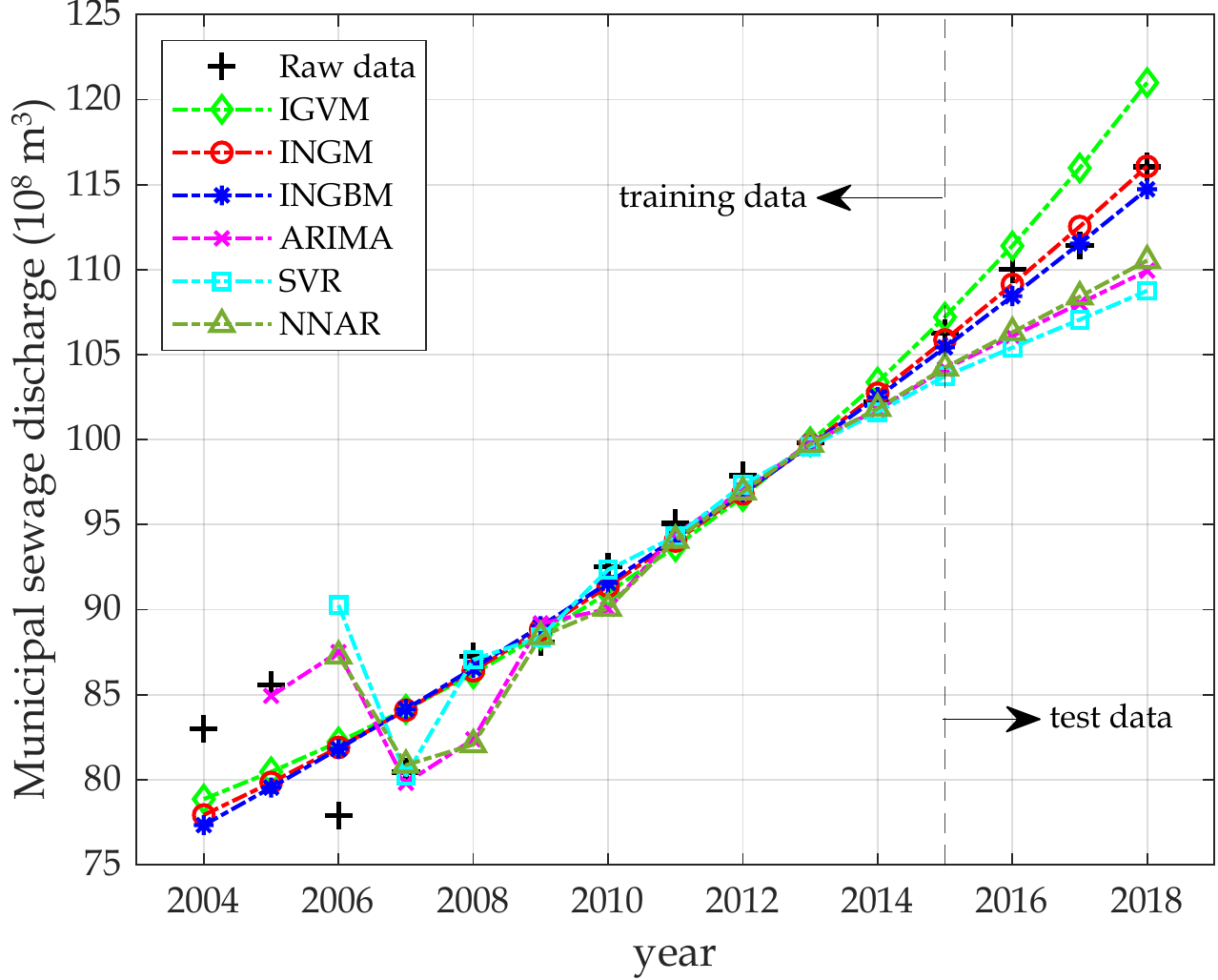}
		\end{minipage}%
	}%
	\subfigure{
		\begin{minipage}[t]{0.5\linewidth}
			\centering
			\includegraphics[width=0.86\linewidth]{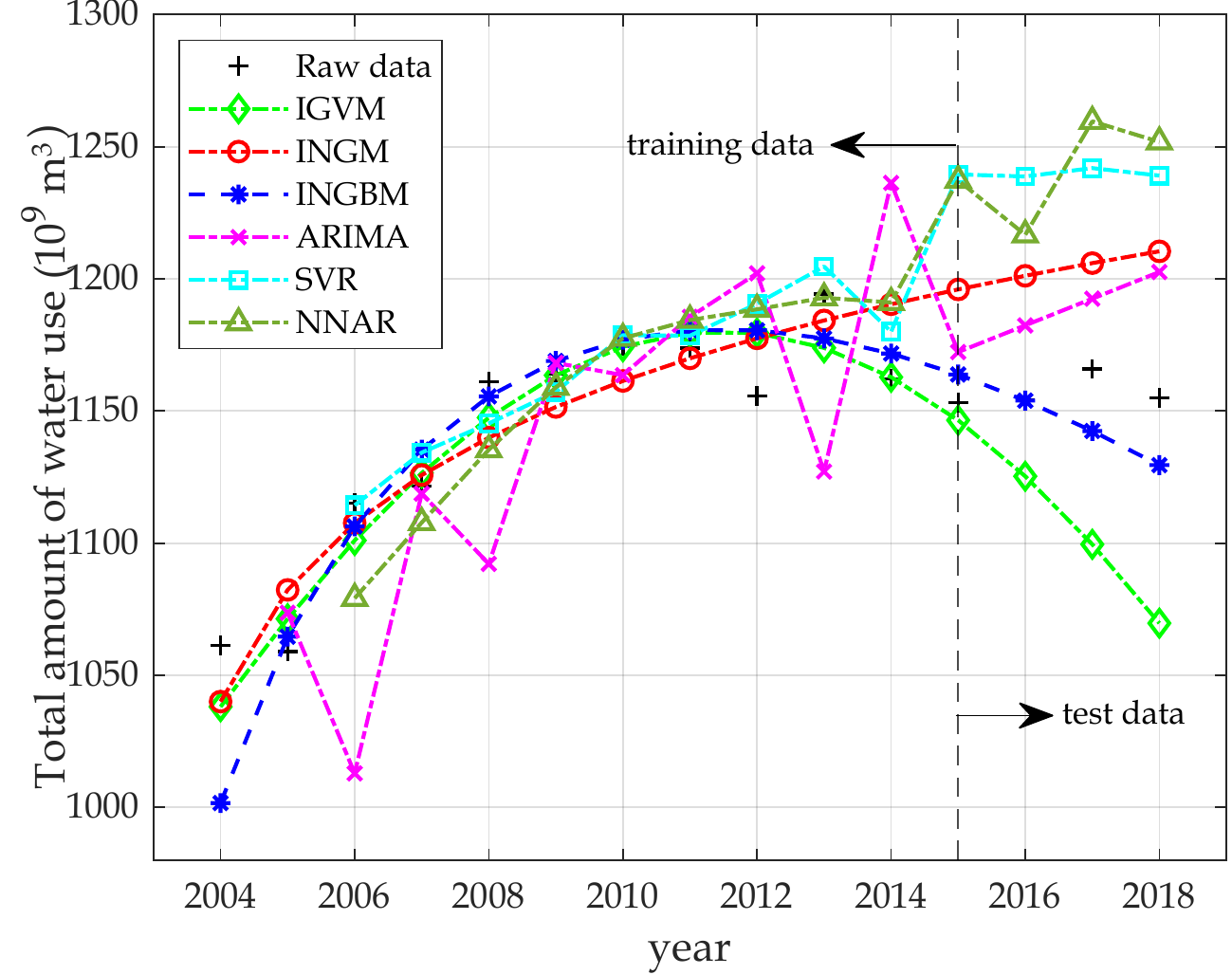}
		\end{minipage}%
	}%
	\caption{Training and testing results for the YRD yearly municipal sewage discharge of and total amount of water use.}
	\label{sewage}
\end{figure}

\subsection{Comparison with other methods}
In addition, the reconstructed nonlinear grey models are compared with time series analysis methods, including autoregressive integrated moving average (ARIMA), support vector regression (SVR), and neural network autoregression (NNAR).
SVR and NNAR models are implemented by $\mathsf{svm}$ and $\mathsf{neural\_network}$ functions in $\mathsf{sklearn}$ package \citep{pedregosa2011scikit}, and AIRMA by the $\mathsf{arima\_model}$ function in $\mathsf{statsmodels}$ package \citep{seabold2010statsmodels} in Python software.

Tables \ref{result_sewage}-\ref{result_consumption} and Figure \ref{sewage} depict the comparison between integro-differential equation models and aforementioned methods.
The results show that, for both cases, the three comparison methods has satisfactory performance on fitting data with 2.55\%, 2.12\% and 2.63\% MAPE$_\mathsf{train}$s in the first case (1.62\%, 1.04\% and 0.77\% MAPE$_\mathsf{train}$s in the second case), which is almost the same as the proposed nonlinear grey models.
But their MAPE$_\mathsf{test}$s vary a lot.
Obviously, in the first case, SVR has the best performance in fitting data with MAPE$_\mathsf{train}$=2.12\%, whereas performs poorly in forecasting with MAPE$_\mathsf{test}$=4.79\%.
In the case of water consumption, the MAPE$_\mathsf{train}$ and MAPE$_\mathsf{test}$ of SVR are 1.04\% and 7.05\% respectively, indicating high fitting accuracy and worse forecasting accuracy.
With regard to NNAR, this result occurs again, which may be over-fitting due to the small data scale.
ARIMA considers the autocorrelation, thereby having satisfactory fitting and forecasting performance, although not the optimal one.
Overall, the INGBM model can be considered as the most appropriate one in the water resources of the YRD cases.

\subsection{Short discussion}
By employing the INGBM model, we calculate three-step ahead forecasting results of municipal sewage discharge and water consumption as \{118.01, 121.38, 124.85\} and \{1115.4, 1100.2, 1084.2\}, respectively.
The results show that the total amount of water use in the YRD is expected to decline in the next few years, which is in accordance with the national water control policy. 
However, with the development of industrialization, the municipal sewage discharge will continue to grow from 2019-2021.
Therefore, the government should constantly enhance the municipal sewage treatment capacity from these aspects: (\romannumeral1) expand the scale of sewage treatment plants to relieve the pressure of sewage treatment system, and improve the standard of sewage treatment and reuse;
(\romannumeral2) constantly update the sewage treatment equipment and sewage treatment technology to increase the efficiency of sewage treatment;
(\romannumeral3) further improve the supervision and management, pay attention to improve the management ability and strengthen the real-time supervision ability.

\section{Conclusions} \label{sec7}

In this paper, we revisit nonlinear grey system models with an integro-differential equation, including unification, reconstruction, and application.
We propose a unified nonlinear grey system model, which covers but not limited to the existing single-variable, multi-variable variable and multi-output models, making it easier for researchers to perform property analysis.
It has shown that the unified framework can be reconstructed to an equivalent integro-differential equation, and an integral operator bridges the two models.
It has shown that the structural parameters and initial values can be estimated simultaneously by the integral matching approach.
The unified framework and a reduced-order integro-differential equation-based model both concern the dynamic modelling, where the former model the Cusum series implicit, whereas the later provide a direct modelling paradigm.
The large-scale simulations demonstrate that the reconstructed model is superior to the original ones, from estimation accuracy to robustness to noise.
In the practical application, data sets of the municipal sewage discharge and the total amount of water use of the YRD serve to validate the reconstructed models in forecasting for real data.

The present work opens up a new way for the development of nonlinear grey system models,
and there are several interesting directions expanding this work:
The unified representation \eqref{gm} is used on the assumption with model structural is known in advance.
Determination of model structural from time series in the presence of measurement error is an interesting topic in the future.
The modelling mechanism of nonlinear grey models has been invaginated from the perspective of mathematics, however, the superiority of which needs to be further probe in physical analysis and practical applications.

\appendix
\section*{Appendix}
\section{Proof for Lemma \ref{lem2}} \label{app1}
\begin{proof}
	Our first destination is to show the case of $d=1$.
	Denoting the parameter vector as
	$\begin{bmatrix}
		{\eta} ~|~
		{\theta}_{\text{L}} ~ | ~
		\bm{\theta}_{\text{N}} ~
	\end{bmatrix} =
	\begin{bmatrix}
		\eta ~|~ a ~ | ~ b_1 ~\cdots~b_{p}
	\end{bmatrix}$,
	the left-hand side of equation \eqref{im2} can be expressed as a algebraic form
	\begin{equation}\label{lm2.1}
		\eta + a \tilde{x}(t_k) + 
		\sum_{m=1}^{p+1}b_m\left[ \left( \eta + \tilde{x}(t_k) \right)^m  - \eta^m\right] 
	\end{equation}
	then, according to the binomial theorem, the polynomial term can be expanded into 
	\[
	\left( \eta + \tilde{x}(t_k)\right) ^m  =
	\sum_{i=1}^{m} \binom{m}{i} \eta^{m-i} \tilde{x}(t_k)^i
	\]
	and thus the problem \eqref{lm2.1} can be manipulated to
	\[
	\eta + a\tilde{x}(t_k) + 
	b_1 \sum_{i=1}^{2} \binom{2}{i}\eta^{2-i} \tilde{x}(t_k)^{i} +
	\cdots +
	b_p\sum_{i=1}^{p+1} \binom{p+1}{i}\eta^{p+1-i} \tilde{x}(t_k)^{i}
	\]
	which can be represented as a matrix form
	\begin{align}\notag
		\begin{bmatrix}
			\eta ~ a ~ b_1 ~ \cdots ~ b_p
		\end{bmatrix}
		\left[
		\begin{array}{c:c:ccc}
			1 & 0 & 0 & \cdots & 0 \\
			\hdashline
			0 & 1  & 0 & \cdots & 0\\
			\hdashline
			0 & \binom{2}{1}\eta^1 & \binom{2}{2}\eta^0 & & \\
			\vdots & \vdots & \vdots & \ddots & \\
			0 & \binom{p+1}{1}\eta^{p} & \binom{p+1}{2}\eta^{q-2}
			& \cdots	& \binom{p+1}{p+1}\eta^{0}
		\end{array}
		\right] 
		\begin{bmatrix}
			1 \\
			\tilde{x}(t_k)\\
			\tilde{x}(t_k)^2 \\
			\vdots \\
			\tilde{x}(t_k)^q
		\end{bmatrix}.   \notag
	\end{align} 
	Denoting $
	\bm{\phi} = \begin{bmatrix}
		\binom{2}{1} \eta^{1} \\
		\vdots \\
		\binom{p+1}{1} \eta^{p}
	\end{bmatrix},
	~
	\text{and}
	~
	\bm{\varphi} = \begin{bmatrix}
		\binom{2}{2} \eta^{0} & & \\
		\vdots & \ddots & \\
		\binom{p+1}{2} \eta^{q-2} & \cdots & \binom{p+1}{p+1} \eta^{0}
	\end{bmatrix}	
	$ gives the equation \eqref{im2}.
	
	Second, when $d \ge 2$, since $\bm{N}(\cdot)$ consists of quadratic nonlinearities, it follows that $p=\frac{d(d+1)}{2}$ and, each component of $\bm{N}(\cdot)$ shares a same form 
	\[
	\left( \eta_i + x_i \right) \left( \eta_j + x_j \right) - \eta_i \eta_j= 
	\eta_i x_j + \eta_j x_i + x_i x_j,
	~i,j=1,2,\cdots,d.
	\]
	
	Accordingly, the second term at the right-hand side in equation \eqref{im.dis} is transformed into
	\begin{equation}\label{eq24}
		\bm{N}(\bm{\eta} + \tilde{\bm{x}}(t_k)) - \bm{N}(\bm{\eta}) = 
		\bm{\psi}\tilde{\bm{x}}(t_k) + \bm{N}(\tilde{\bm{x}}(t_k))
	\end{equation}
	where $\bm{\psi}=\begin{bmatrix}
		\psi_d^\top & \psi_{d-1}^\top & \cdots & \psi_1^\top
	\end{bmatrix}^\top$,
	for
	\[
	\psi_{d-i} =  
	\underset{i \text{ columns}}{\underbrace{
			\left[ 
			\begin{array}{cccc}
				0 		& \cdots & 0		\\
				0 		& \cdots & 0		\\
				\vdots  & \ddots & \vdots		\\
				0		& \cdots & 0		\\
			\end{array}
			\right.
	}}
	\;
	\underset{d-i \text{ columns}}{\underbrace{
			\left.
			\begin{array}{cccc}
				2\eta_{i+1}		&				&		&	0	\\
				\eta_{i+2}		&\eta_{i+1}		&		&		\\
				\vdots 			&				& \ddots&		\\
				\eta_d 			&0				&		&	\eta_{i+1}	\\
			\end{array}
			\right]
	}},
	~
	i = 0,~1,\cdots,~d-1
	\]
	
	Substituting equation \eqref{eq24} into problem \eqref{im.dis} gives equation \eqref{im3}.	
\end{proof}	
	
\section{Modelling results for a multi-output model} \label{append1}

This simulation example is designed to further evaluate the proposed method under a high-dimensional setting.
We adopt the reduced nonlinear grey Lotka-Volterra model:
\[
\begin{cases}
	\frac{d}{dt} x_1 = a_{1}x_1 - b_{1}\left[ x_1\left( \eta_2 + \int_{t_1}^{t} x_2(\tau) d\tau \right) + x_2\left(\eta_1 + \int_{t_1}^{t} x_1(\tau) d\tau \right)  \right]   \\
	\frac{d}{dt} x_2 = a_{2}x_2 - b_{2}\left[ x_1\left( \eta_2 + \int_{t_1}^{t} x_2(\tau) d\tau \right) + x_2\left(\eta_1 + \int_{t_1}^{t} x_1(\tau) d\tau \right)  \right]  \\
\end{cases},
~
\begin{array}{l}
	x_1(t_1) = \eta_1 \\
	x_2(t_2) = \eta_2
\end{array}
\]
Correspondingly, nonlinear grey Lotka-Volterra model is given by
\[
\begin{cases}
	\frac{d}{dt} y_1 = a_{1} y_1 - b_{1} y_1 y_2 \\
	\frac{d}{dt} y_2 = a_{2} y_2 - b_{2} y_1 y_2
\end{cases},
~\begin{array}{l}
	y_1(t_1) = \eta_1\\
	y_2(t_1) = \eta_2
\end{array}
\]

Here, we set the structural parameters to $a_{1} = 1.2, ~b_{1} = 0.3, ~a_{2} = -1.0$ and $b_{2} = -0.4$, the initial value to $\eta_1 = 5.0$ and $\eta_2 = \frac{2}{3}$.

Experiments are conducted in the following fashion.
Initially, let $T = 5$ and $h=0.01$, a number of $n = 501$ samples are generated, following that the noise level $[4\%,~8\%,~12\%,~16\%]$ are considered.
In turn, later, we conduct the experiment with a fixed noise magnitude $4\%$ and varying sample size $ n = [21,~51,~101,~501]$.
Simulation results are summarised Figures \ref{lot_nlen} and \ref{lot_nois}.

\begin{figure}[h]
	\centering
	\includegraphics[width=1\linewidth]{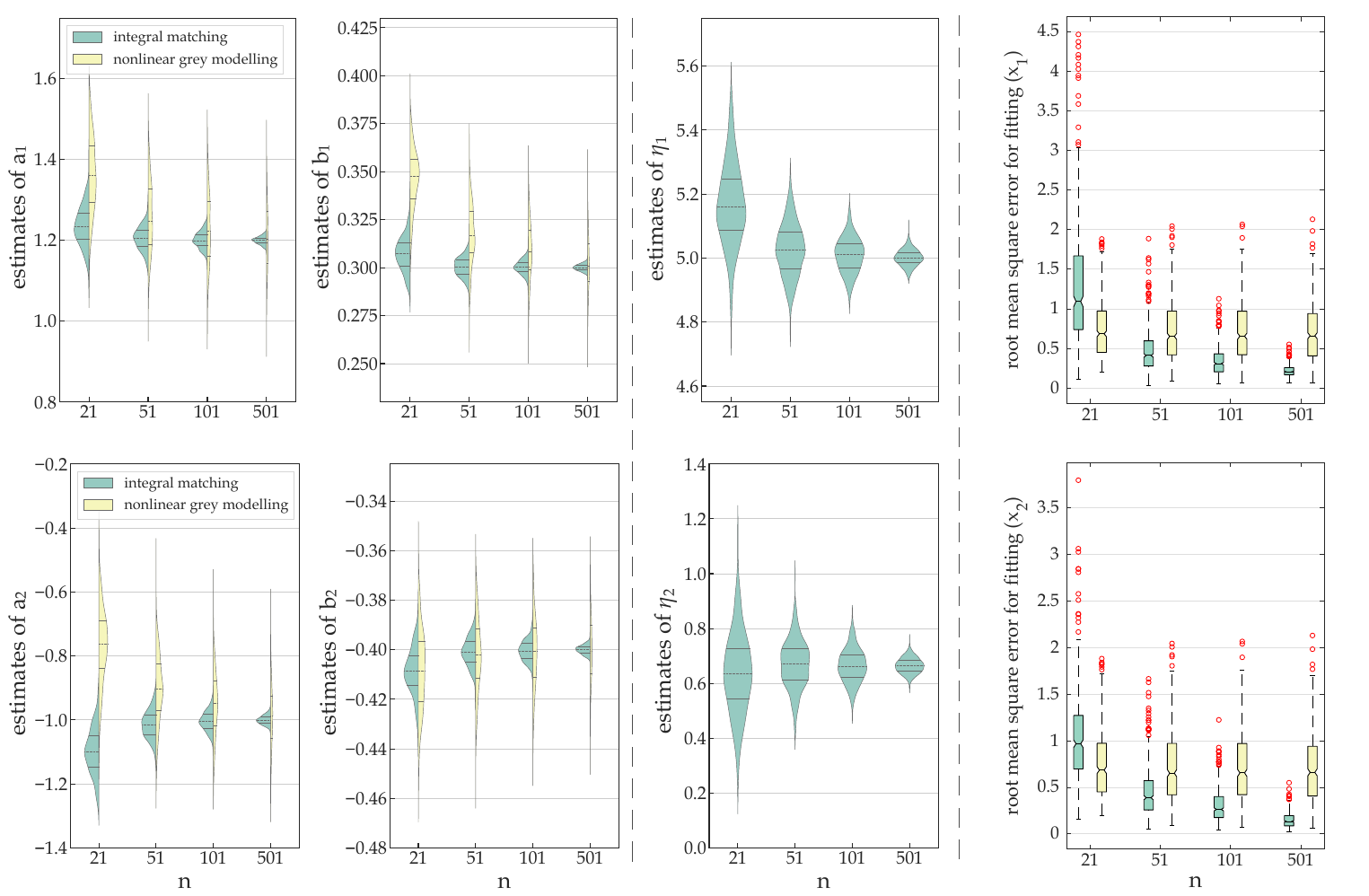}
	\caption{Violin plots of estimated structural parameters, initial value and boxplot of fitting error.
		The true parameters and initial value are $a_{1}=1.2,~b_{1}=0.3,~a_{2}=-1,~b_{2}=-0.4$ and $\eta_1=5,~\eta_2=2/3$.
		Note, nonlinear grey modelling uses the true values of the initial condition $\eta_1=5$ and $\eta_2=2/3$ for solving the time response function.}
	\label{lot_nlen}
\end{figure}
\begin{figure}[h]
	\centering
	\includegraphics[width=1\linewidth]{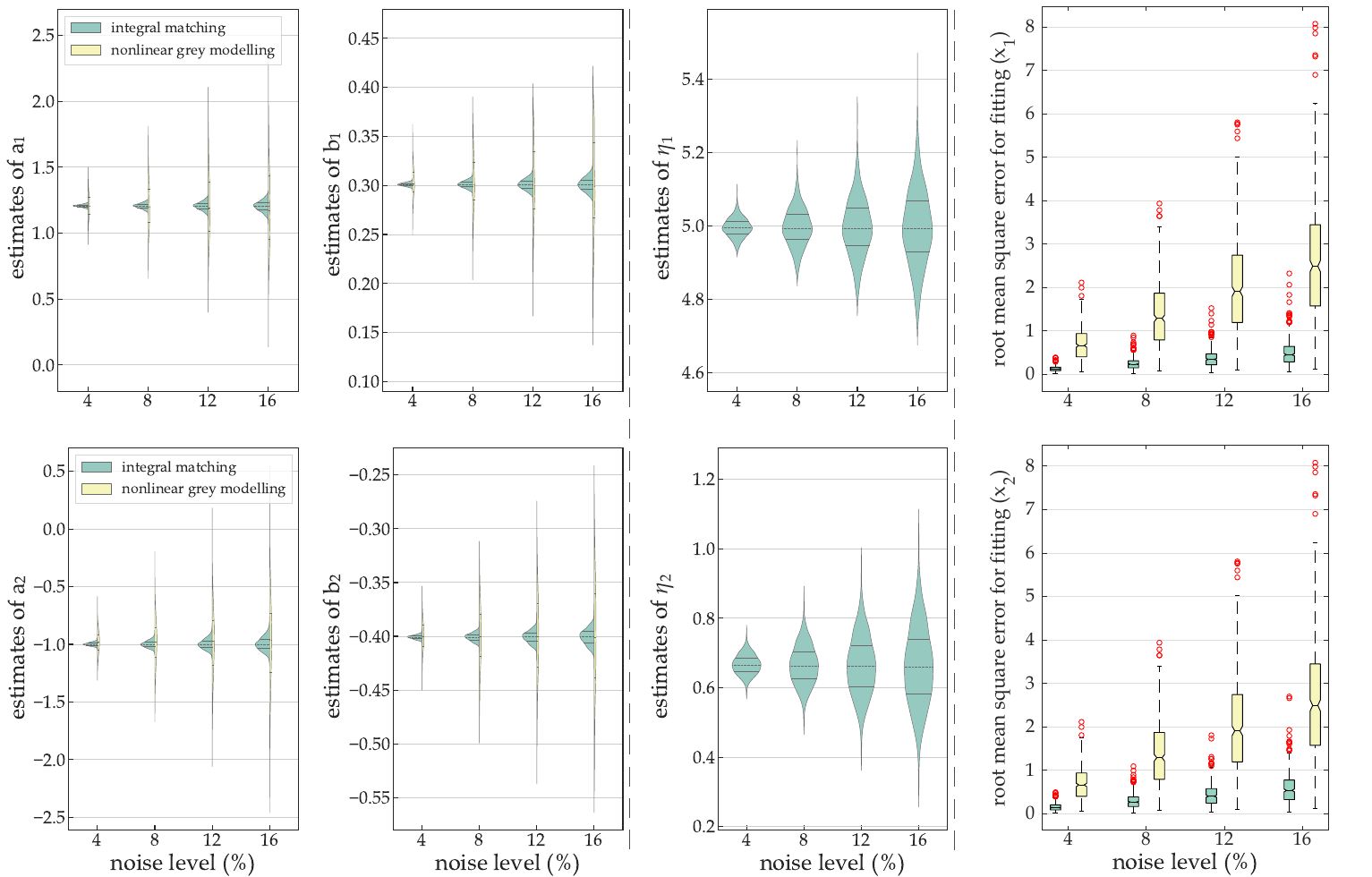}
	\caption{Violin plots of estimated structural parameters, initial value and boxplot of fitting error. The true parameters and initial value are $a_{1}=1.2,~b_{1}=0.3,~a_{2}=-1,~b_{2}=-0.4$, and $\eta_1=5,~\eta_2=2/3$.
		Note, nonlinear grey modelling uses the true values of the initial condition $\eta_1=5$ and $\eta_2=2/3$ for solving the time response function.}
	\label{lot_nois}
\end{figure}

\section*{Acknowledgment}
This work was supported by the National Natural Science Foundation of China (72171116) and the Fundamental Research Funds for the Central Universities of China (NP2020022).

\section*{Data accessibility}
Data and codes are available at \url{https://github.com/Yanglu0319/Integral-NGM}.

\printcredits

\bibliographystyle{cas-model2-names}

\bibliography{NGBM}

\end{document}